%% file: main.tex
\documentclass[english]{svjour3}
\usepackage{url}
\usepackage{amsmath}
\usepackage{amssymb}

\makeatletter
\RequirePackage{fix-cm}

\smartqed
\usepackage{tikz}
\usepackage{hyperref}	
\usepackage{amsmath}    
\usepackage{amssymb}    
\usepackage{enumitem} 
\usepackage{subcaption} 
\usepackage{booktabs}   
\usepackage{pgfplots}	
\usepackage{setspace}   

\hypersetup{
colorlinks = true,
linkcolor  = blue,
citecolor  = blue,
urlcolor   = blue }

\newcommand{\tableTop}[3] {
\begin{table}[#1]
    #3
    \centering
    \begin{tabular}
        #2
    \end{tabular}
\end{table}
}

\usepackage{babel}

\begin{document}
\title{Zero-Error Tracking for Autonomous \\
Vehicles through Epsilon-Trajectory Generation}
\subtitle{}
\author{Clint Ferrin \and Greg Droge \and  Randall Christensen}
\institute{C. Ferrin \at  Utah State University, Logan, Utah 84322 USA\\
Tel.: +1-435-554-8176\\
\email{clint.ferrin@gmail.com}\\
}
\maketitle
\begin{abstract}
This paper presents a control method and trajectory planner for vehicles
with first-order nonholonomic constraints that guarantee asymptotic
convergence to a time-indexed trajectory. To overcome the nonholonomic
constraint, a fixed point in front of the vehicle can be controlled
to track a desired trajectory, albeit with a steady-state error. To
eliminate steady state error, a sufficiently smooth trajectory is
reformulated for the new reference point such that, when tracking
the new trajectory, the vehicle asymptotically converges to the original
trajectory. The resulting zero-error tracking law is demonstrated
through a novel framework for creating time-indexed Clothoids. The
Clothoids can be planned to pass through arbitrary waypoints using
traditional methods yet result in trajectories that can be followed
with zero steady-state error. The results of the control method and
planner are illustrated in simulation wherein zero-error tracking
is demonstrated.

\keywords{nonholonomic systems, motion planning, feedback control, autonomous
vehicles, mobile robots, dynamic partial feedback linearization.} 
\end{abstract}

\input{content.tex}

\end{document}

%% file: content.tex
\section{Introduction}

The control of mobile robots continues to be a focus of academic and
industrial research, especially due to a demand for high-precision
path following for a variety of mobile robots. Though the dynamic
equations for mobile robots can vary dramatically, they share similar
trajectory-tracking limitations due to their nonholonomic constraints
and nonlinear equations of motion that warrant continued research
for vehicle control and path planning e.g. \cite{Kapitanyuk2018}.

One limitation of vehicles with first-order nonholonomic constraints
is that they cannot move orthogonal to their direction of motion \cite{DeLuca2001Chapter7M}.
This limitation is overcome through various approaches including open-loop
control \cite{Murray1993}, linearizing about a trajectory \cite{DeLuca1998},
or reformulating the control problem using a path-following controller
that is not time-dependent \cite{Aguiar2005}. The most common approach
is to use a path-following controller which separates the system into
two separate controllers\cite{Aguiar2005,Kapitanyuk2018}. One controller
governs the vehicle's velocity, while the other controller uses the
vehicle's steering input to converge to the desired path. Though the
path-following controller addresses the issue of controllability,
new complexities arise such as how to treat self-intersecting paths,
and how to optimally calculate which point on the path to use as a
reference because the problem is not time-dependent.

An alternative approach is to control a point on the robot that does
not have a nonholonomic constraint. Instead of controlling the center
point of a fixed axle, \cite{Olfati-Saber} showed that a point $\epsilon$
in front could be controlled as if it were unconstrained, allowing
for partial feedback linearization. The inputs of the unconstrained
control can then be mapped algebraically to the control inputs of
the vehicle, and the reference trajectory is tracked with a steady-state
error of $\epsilon$. This control technique is referred to as $\epsilon$-tracking,
and the theoretical limit for $\epsilon$ is solely that it remains
positive, allowing the reference tracking error to become arbitrary
small \cite{Olfati-Saber}. In practice, however, very small values
of $\epsilon$ can become problematic due to noise introduced by dynamics,
sensors, or even numeric integration. Therefore, the first major contribution
of this paper is to generate a new trajectory that, when followed
using $\epsilon$-tracking, results in zero-error, asymptotic tracking
of the vehicle control point to the original reference trajectory.
The new tracking law is referred to as zero-error $\epsilon$-trajectory
tracking.

Another common limitation for mobile vehicles that must be taken into
account when developing trajectories for vehicles is the executability
of a trajectory. The maximum curvature limitation for path planning
is commonly addressed with the Dubins path, which is formed using
a waypoint path planner that produces optimal distance paths between
two oriented points using circular arcs and straight lines while taking
into account maximum curvature constraints \cite{Dubins1957}. Dubins
paths, however, require instantaneous changes in curvature which is
often not achievable in physical systems due to actuator limitations
and wheel slippage \cite{Meidenbauer2007}. Fraichard and Scheuer
proposed adapting the Dubins path by connecting waypoints using straight
lines, circular arcs, and transition arcs to account for the maximum
change of curvature \cite{Fraichard2004}, which are referred to as
continuous curvature paths (CCPaths). The addition of the transition
arc maintains curvature continuity between circles and lines. The
second contribution of this research is to convert CCPaths paths into
trajectories with corresponding feed-forward terms which can be used
for control. The result of these contribution is a path planner that
connects waypoints to form a time-dependent trajectory that can be
tracked by vehicles with first-order nonholonomic constraints with
guaranteed asymptotic convergence.

The remainder of this paper begins in Section \ref{sec:Background}
by outlining the necessary background for the development of the proposed
zero-error $\epsilon$-trajectory tracking. Section \ref{sec:Zero-Error-Epsilon-Tracking}
presents the main result of the paper: error-free trajectory tracking
accomplished by redefining an arbitrary reference trajectory in terms
of the $\epsilon$-point. Section \ref{sec:CCPath-Tracking} presents
an example application by extending continuous curvature paths to
trajectories and using the $\epsilon$-trajectory tracking method
to follow a desired trajectory. The paper ends with concluding remarks
in Section \ref{sec:Conclusion}.

\section{Preliminaries\label{sec:Background}}

This section presents the preliminary information necessary to understand
the contributions of this work. The motion models used for vehicle
control and analysis are first developed, followed by an introduction
to the $\epsilon$-point control. This section ends with a discussion
of the CCPaths.

\subsection{Kinematic Models\label{subsec:Kinematic-Models}}

Four kinematic models are presented in this section. The first two
describe the vehicle motion models being controlled, the third describes
a model used to create trajectories, and the last is a model used
solely for convergence analysis. The simplest model to be employed
is the unicycle model. It is a model of a first-order nonholonomic
constrained system that can be represented as 
\begin{align}
\left[\begin{array}{c}
\dot{x}\\
\dot{y}\\
\dot{\psi}\\
\dot{v}\\
\dot{\omega}
\end{array}\right] & =\left[\begin{array}{c}
v\cos\psi\\
v\sin\psi\\
\omega\\
0\\
0
\end{array}\right]+\left[\begin{array}{cc}
0 & 0\\
0 & 0\\
0 & 0\\
1 & 0\\
0 & 1
\end{array}\right]\left[\begin{array}{c}
a\\
\alpha
\end{array}\right],\label{eq:unicycle_dynamics}
\end{align}
where $\left(x,y\right)$ is the position of the robot, $\psi$ is
the heading, $v$ is the translational velocity, $\omega$ is the
rotational velocity, and $a$ and $\alpha$ are the longitudinal and
angular acceleration, respectively, e.g. \cite{LaValle2006}. The
positions, velocities, and accelerations can be grouped together as
$\underbar{x}=\left[x,y\right]^{T}$, $\underbar{v}=\left[v,\omega\right]^{T}$,
and $\underbar{a}=\left[a,\alpha\right]^{T}$. The first two rows
in (\ref{eq:unicycle_dynamics}) describe the lateral motion constraint
of a wheeled vehicle. As a result, more complex kinematic models can
be related to (\ref{eq:unicycle_dynamics}) using algebraic mappings
\cite{LaValle2006}.

One such model that will be used is the Ackermann-style bicycle model
\cite{siegwart2011}. It represents the lateral motion constraint
as in (\ref{eq:unicycle_dynamics}) with the angular velocity being
a function of the steering angle and longitudinal velocity to present
a more realistic input for steered vehicles. It can be represented
as 
\begin{align}
\left[\begin{array}{c}
\dot{x}\\
\dot{y}\\
\dot{\psi}\\
\dot{v}\\
\dot{\phi}
\end{array}\right] & =\left[\begin{array}{c}
v\cos\psi\\
v\sin\psi\\
\frac{v}{L}\tan\phi\\
0\\
0
\end{array}\right]+\left[\begin{array}{cc}
0 & 0\\
0 & 0\\
0 & 0\\
1 & 0\\
0 & 1
\end{array}\right]\left[\begin{array}{c}
a\\
\xi
\end{array}\right],\label{eq:ackermann_dynamic_state_equations}
\end{align}
where $\phi$ is the steering angle, $L$ is the wheelbase, and $\xi$
is the steering rate, as seen in Figure \ref{fig:trailer_dynamics}.
This model is used in many kinematic controllers because $\xi$ can
be directly related to a steering wheel input.

The extended Dubins model, which was initially introduced in \cite{boissonnat1994}
to explore the distance optimal path for trajectories in $\mathbb{C}^{2}$,
where $\mathbb{C}^{k}$ denotes the set of functions which are $k$
times continuously differentiable, proves useful for creating trajectories
which respect the curvature constraints of the system. It can again
be seen as an extension to (\ref{eq:unicycle_dynamics}) where the
angular velocity is calculated in terms of the curvature as follows
\begin{align}
\left[\begin{array}{c}
\dot{x}\\
\dot{y}\\
\dot{\psi}\\
\dot{\kappa}
\end{array}\right] & =\left[\begin{array}{c}
v\cdot\cos\psi\\
v\cdot\sin\psi\\
v\cdot\kappa\\
0
\end{array}\right]+\left[\begin{array}{c}
0\\
0\\
0\\
1
\end{array}\right]\sigma,\label{eq:extended_dubins}
\end{align}
where $\kappa$ is curvature and $\sigma$ is the change in curvature.
There are no dynamics for $v$ as it is assumed constant.

To prove convergence of the $\epsilon$-tracking algorithm, a trailer
model will be used. The trailer model can be viewed as an extension
to (\ref{eq:unicycle_dynamics}) where a trailer is connected via
a hitch to the rear axle, as depicted in Figure \ref{fig:trailer_dynamics}.
The trailer model is given by 
\begin{align}
\left[\begin{array}{c}
\dot{x}\\
\dot{y}\\
\dot{\psi}\\
\dot{\psi}_{t}\\
\dot{v}\\
\dot{\omega}
\end{array}\right] & =\left[\begin{array}{c}
v\cos\psi\\
v\sin\psi\\
\omega\\
\frac{v}{d}\sin\left(\psi-\psi_{t}\right)\\
0\\
0
\end{array}\right]+\left[\begin{array}{cc}
0 & 0\\
0 & 0\\
0 & 0\\
0 & 0\\
1 & 0\\
0 & 1
\end{array}\right]\left[\begin{array}{c}
a\\
\alpha
\end{array}\right]\label{eq:trailer_equations}
\end{align}
where $\psi_{t}$ is the heading of the trailer, and $d$ is the length
of the hitch (e.g. \cite{Murray1993}). Note that the position of
the trailer can be directly computed from $\underbar{x}$, $d$, and
$\psi_{t}$ and is therefore not included as a state in (\ref{eq:trailer_equations}).

\subsection{Tracking Methods\label{subsec:Tracking-Methods}}

With the kinematic models to be controlled in hand, our attention
now turns to the point-control laws that will be extended to provide
exact tracking of a desired trajectory. A myriad of control approaches
have been used for tracking a desired path such as utilizing approximate
linearization \cite{DeLuca1998}, sinusoid control inputs \cite{Murray1993},
vector field following \cite{Kapitanyuk2018}, or other nonlinear
control laws \cite{Hoffmann2007,Snider2009,Soltesz2008} to name a
few.

This section focuses on the approach developed by Olfati-Saber as
it provides a framework to directly deal with the nonholonomic constraints
in wheeled vehicles by focusing on a point directly in front of the
vehicle, a point referred to as the $\epsilon$-point \cite{Olfati-Saber}.
This is, of course, not the only method that has used a point in front
of the vehicle for the basis of control. For example, the pure-pursuit
method \cite{Amidi1990,Coulter1992} uses a look-ahead point to connect
the vehicle to the desired path using a proportional curvature controller.
As the look-ahead point increases, the robot experiences greater stability
but results in steady-state errors and requires tuning for different
speeds \cite{Park2014}. The $\epsilon$-point tracking introduced
by Olfati-Saber has a similar steady-state error trend without the
need for extensive tuning due to the global exponential stability
guarantees.

In what follows, a simplified version of the control presented by
Olfati-Saber is presented that allows for partial feedback linearization.
The results of this section can be derived from the results in \cite{Olfati-Saber}
and are presented for the sake of clarity in understanding both the
notation and contributions of future sections. The simplification
of this section with respect to \cite{Olfati-Saber} is that $\epsilon$
is maintained constant whereas, in \cite{Olfati-Saber}, $\epsilon$
decreases to an arbitrarily small value. Theoretically, this arbitrary
smallness is beneficial for tracking a trajectory; practically, disturbances
pose issues for implementation.

Consider a point $\underbar{q}=\left[x,y\right]^{T}\in\mathbb{R}^{2}$
where $\ddot{\underbar{q}}$ is the control input to the system. Let
$\underbar{p}=\left[\underbar{q},\dot{\underbar{q}}\right]^{T}$.
The dynamics of the point can be written as 
\begin{align}
\dot{\underbar{p}}=\mathbf{A}\underbar{p}+\mathbf{B}\underbar{u}=\left[\begin{array}{cc}
\mathbf{0} & \mathbf{I}\\
\mathbf{0} & \mathbf{0}
\end{array}\right]\left[\begin{array}{c}
\underbar{q}\\
\dot{\underbar{q}}
\end{array}\right]+\left[\begin{array}{c}
\mathbf{0}\\
\mathbf{I}
\end{array}\right]\underbar{u}\label{eq:point-control-dynamics}
\end{align}
where $\mathbf{I}$ is the $2\times2$ identity matrix and $\mathbf{0}$
is a $2\times2$ matrix of zeros. Given the proposed system, the point
$\underbar{p}$ exponentially tracks a reference trajectory per the
following lemma.


    \begin{figure}[t]
        \centering
        \normalsize
        \def\svgwidth{.75\linewidth}
        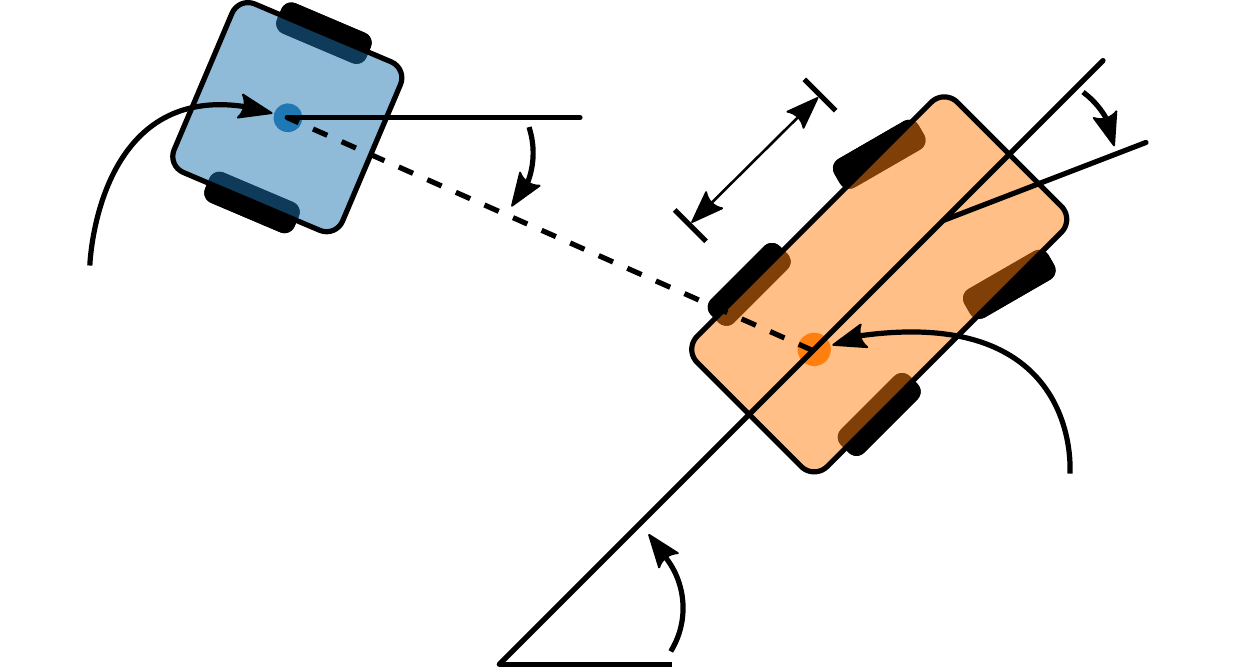
        \caption{Dynamics of a car pulling a trailer where $L$ is the wheelbase, $\phi$ is the steering angle of the pulling vehicle, and $d$ is the length of the hitch between the connected trailer and the pulling vehicle. The angular velocity of the pulling vehicle can be computed using $\phi$ and $L$ where $\omega=\frac{v}{L}\tan\phi$.}
        \label{fig:trailer_dynamics}
    \end{figure}

\begin{lemma}
\label{lem:expo-point-control}Given the system defined in (\ref{eq:point-control-dynamics}),
let $\underbar{q}_{r}=\left[x_{r},y_{r}\right]^{T}\in\mathbb{C}^{2}$
, $\underbar{p}_{r}=\left[\underbar{q}_{r},\dot{\underbar{q}}_{r}\right]^{T}$
and $\underbar{u}=\ddot{\underbar{x}}_{r}-K\left(\underbar{p}-\underbar{p}_{r}\right)$,
with the real parts of the eigenvalues of $\mathbf{A}-\mathbf{B}K$
being negative, then $\underbar{q}\left(t\right)$ is globally exponentially
stable to $\underbar{x}_{r}\left(t\right)$.
\end{lemma}
\begin{proof}
Let an error state $\underbar{z}$ be introduced where $\underbar{z}=\underbar{p}-\underbar{p}_{r}$.
Defining $\underbar{u}=-K\underbar{z}+\ddot{\underbar{q}}_{r}$, the
time derivative of $\underbar{z}$ can be written as $\dot{\underbar{z}}=(A-BK)\underbar{z}$.
It can be verified that $(A,B)$ forms a completely controllable linear
time-invariant system which completes the proof.
\end{proof}

Lemma \ref{lem:expo-point-control} cannot be directly applied to
the unicycle model as the inputs and dynamics are obviously not matched.
However, $\underbar{q}_{\epsilon}$ can be defined as a point directly
in front of the vehicle such that 
\begin{align}
\underbar{q}_{\epsilon} & =\begin{bmatrix}x\\
y
\end{bmatrix}+\epsilon\begin{bmatrix}\cos\psi\\
\sin\psi
\end{bmatrix}\label{eq:control_ref_position}
\end{align}
where $\epsilon>0$ is constant. The point $\underbar{q}_{\epsilon}$
can then be defined so that it behaves as if it were the constraint-free
point in Lemma \ref{lem:expo-point-control} using an algebraic relationship
between $\underbar{u}$ in Lemma \ref{lem:expo-point-control} and
the control inputs in (\ref{eq:unicycle_dynamics}). This relationship
is stated in the following lemma.
\begin{lemma}
\label{lem:converge_unicycle}Let $\underbar{x}_{r}=\left[x_{r},y_{r}\right]^{T}\in\mathbb{C}^{2}$
be a reference trajectory for a vehicle with unicycle dynamics as
defined in (\ref{eq:unicycle_dynamics}). Let $\underbar{u}_{\epsilon}$
be defined as the control input to the point system from Lemma \ref{lem:expo-point-control}.
Define the system input as: \normalfont
\begin{align}
\underbar{a}=R_{\epsilon}^{-1}\underbar{u}_{\epsilon}-\hat{\omega}\underbar{v}.\label{eq:a_bar_control_input}
\end{align}
\itshape where 
\begin{align}
R_{\epsilon}^{-1} & =\left[\begin{array}{cc}
\cos\psi & -\frac{1}{\epsilon}\sin\psi\\
\sin\psi & \frac{1}{\epsilon}\cos\psi
\end{array}\right]
\end{align}
then $\underbar{q}_{\epsilon}\left(t\right)\rightarrow$ $\underbar{x}_{r}\left(t\right)$
globally and exponentially fast.
\end{lemma}
\begin{proof}
The first and second derivatives for $\underbar{q}_{\epsilon}$ can
be directly calculated and written as: 
\begin{align}
\dot{\underbar{q}}_{\epsilon} & =R_{\epsilon}\underbar{v},\label{eq:derivative_control_input}\\
\ddot{\underbar{q}}_{\epsilon} & =R_{\epsilon}\hat{\omega}\underbar{v}+R_{\epsilon}\underbar{a}
\end{align}
\begin{align}
R_{\epsilon} & =\begin{bmatrix}\begin{array}{cc}
\cos\psi & -\epsilon\sin\psi\\
\sin\psi & \epsilon\cos\psi
\end{array}\end{bmatrix},\ \hat{\omega}=\left[\begin{array}{cc}
0 & -\epsilon\omega\\
\frac{\omega}{\epsilon} & 0
\end{array}\right].\label{eq:eps_matrix_definitions}
\end{align}
Therefore, by setting $\ddot{\underbar{q}}_{\epsilon}$ equal to $\underbar{u}_{\epsilon}$
and algebraically solving for $\underbar{a}$, the dynamics for $\underbar{q}_{\epsilon}\left(t\right)$
are matched to (\ref{eq:point-control-dynamics}), and $\underbar{q}_{\epsilon}\left(t\right)\rightarrow$
$\underbar{x}_{r}\left(t\right)$ globally and exponentially fast,
per Lemma \ref{lem:expo-point-control}.
\end{proof}

\begin{corollary}
Due to Lemma \ref{lem:converge_unicycle}, $\left\Vert \underbar{e}\left(t\right)\right\Vert \rightarrow\epsilon$
where $\underbar{e}=\underbar{x}-\underbar{x}_{r}$.
\end{corollary}
This corollary can be seen by noting that the $\epsilon$-point converges
to $\underbar{x}_{r}$ and, by definition, $\underbar{x}$ is a distance
$\epsilon$ from the $\underbar{q}_{\epsilon}$.

Lemma \ref{lem:converge_unicycle} can be directly extended to the
bicycle model in (\ref{eq:ackermann_dynamic_state_equations}) by
forming one additional algebraic relationship which relates the steering
input of the bicycle to the angular acceleration input of the unicycle.
A small contribution to $\epsilon$-tracking is now made by extending
it to the bicycle model in (\ref{eq:ackermann_dynamic_state_equations}),
as shown in Figure \ref{fig:bicycle_control_point}. In the bicycle
model, the angular velocity is given by $\omega=\frac{v}{L}\tan\phi$.
By differentiating the angular velocity, the angular acceleration
can be written as: 
\begin{align}
\alpha & =\frac{a}{L}\cdot\tan\left(\phi\right)+\frac{v}{L}\cdot\frac{\xi}{\cos^{2}\left(\phi\right)}.\label{eq:acker_angular_accel}
\end{align}
Rearranging terms, the steering input can be expressed as: 
\begin{align}
\xi & =\frac{1}{v}\cos^{2}\left(\phi\right)\left(L\alpha-a\cdot\tan\left(\phi\right)\right)\text{.}\label{eq:bicycle-angular-accel}
\end{align}
Thus, given the inputs $\underbar{a}$ in Lemma \ref{lem:converge_unicycle},
the inputs to the bicycle, $(a,\xi)$, can be directly computed as
long as $v\neq0$.


    \begin{figure}[t]
        \centering
        \normalsize
        \def\svgwidth{.75\linewidth}
        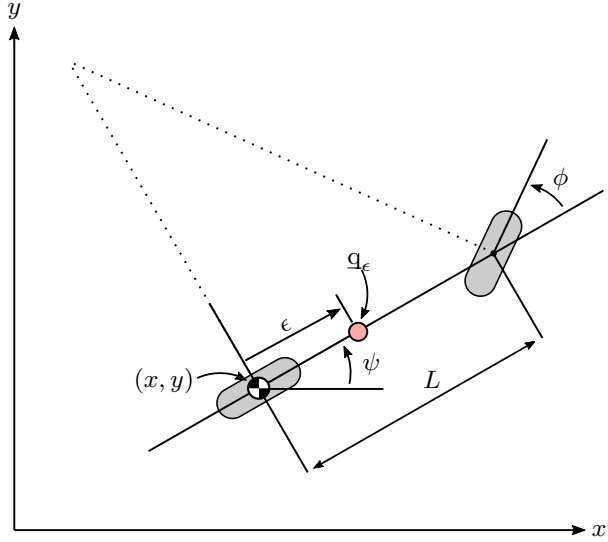
        \caption{Proposed control reference point $\underbar{q}_\epsilon$ used to linearize the bicycle model.}
        \label{fig:bicycle_control_point}
    \end{figure}

\subsection{Continuous Curvature Paths\label{subsec:Waypoint-Following-Clothoids}}

The CCPath was first introduced by Fraichard and Scheuer to plan a
path that accounts for the maximum curvature and maximum change in
curvature \cite{Fraichard2004}. This allows for planning paths that
consider fundamental vehicle execution constraints or even, as shown
in \cite{Villagra2012a}, to adhere to passenger acceleration and
jerk constraints. A CCPath is formed by connecting two oriented waypoints
using straight lines, circular arcs, and transition arcs as seen in
Figure \ref{fig:clothoid_connection}.

The transition arcs are formed by generating a Clothoid, also referred
to as the Euler Spiral, which is defined as a path whose curvature
increases linearly with arc length \cite{Levien2008}. Because the
Clothoid produces continuous curvature paths, it approximates the
kinematic curvature limitations imposed by physical actuators on mobile
robots \cite{Meidenbauer2007}. The parametric equations to form a
Clothoid are referred to as the Fresnel integral and are represented
as 
\begin{align}
\begin{split}\begin{aligned}x=C_{f}=x_{0}+\int_{0}^{s}\cos\left(\frac{1}{2}\sigma(\xi^{2}+\kappa_{0}\xi+\psi_{0})\right)d\xi\\
y=S_{f}=y_{0}+\int_{0}^{s}\sin\left(\frac{1}{2}\sigma(\xi^{2}+\kappa_{0}\xi+\psi_{0})\right)d\xi
\end{aligned}
\end{split}
\label{eq:fresnel_integrals}
\end{align}
where $\left(x,y\right)$ is the position, $\psi$ is the heading,
$\kappa$ is curvature, and $\sigma$ is the change in curvature \cite{Lekkas2014}.

The general case for a turn in a CCPath is referred to as a continuous
curvature turn (CCTurn) which is composed of three stages for both
positive and negative velocity: (1) a transition arc of linearly changing
curvature from 0 to $\pm\kappa_{\text{max}}$ at a rate of $\pm\sigma_{\text{max}}$,
where $\kappa_{\text{max}}$ is the maximum allowable curvature and
$\sigma_{\text{max}}$ is the maximum allowable change in curvature,
(2) a circular arc with a constant curvature of $\pm\kappa_{\text{max}}$,
and (3) a transition arc of linearly changing curvature from $\pm\kappa_{\text{max}}$
to $0$ at a rate of $\pm\sigma_{\text{max}}$ \cite{Fraichard2004}.

One difficulty of path-following controllers for CCPaths is determining
which point on the path to follow. This can be done by finding the
closest lateral point on the path, often referred to as the cross-track
error. Calculating cross-track error introduces new complexities including
how to optimally calculate or measure cross-track error, and how to
avoid errors originating from self-intersecting paths \cite{Kapitanyuk2018}.
Following a time-indexed path (referred to as a trajectory) largely
eliminates this difficulty by using a time-index to determine the
point to follow. Note the similarity of (\ref{eq:fresnel_integrals})
to the unicycle model in (\ref{eq:unicycle_dynamics}). The evolution
of $x$ depends on the cosine of some changing value and the evolution
of y likewise depends on the sine. This will be exploited in Section
\ref{sec:CCPath-Tracking} to create a natural time-indexing of a
CCPath based on the kinematic constraints of a vehicle.


    \begin{figure}[t]
        \centering
        \normalsize
        \def\svgwidth{.75\linewidth}
        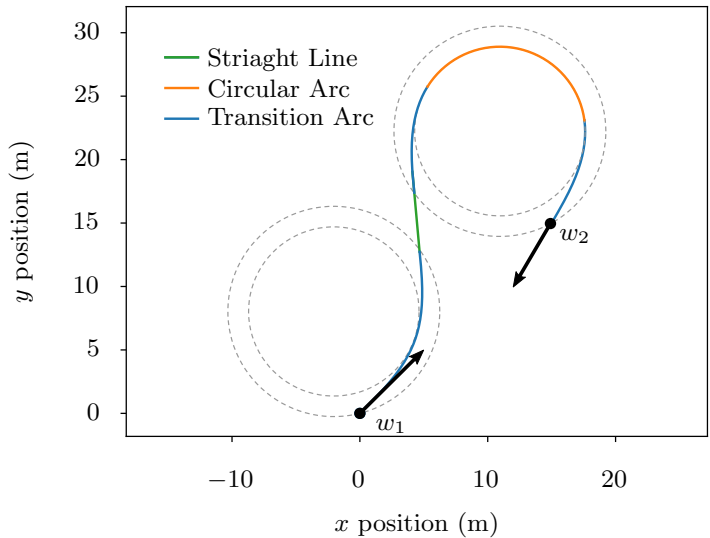
        \caption{Connecting directional waypoints $w_1$ and $w_2$ using straight lines, circular arcs, and transition arcs composed of Clothoids.}
        \label{fig:clothoid_connection}
    \end{figure}

\section{Zero-Error $\epsilon$-Trajectory Tracking\label{sec:Zero-Error-Epsilon-Tracking}}

This section uses the $\epsilon$-tracking method to achieve zero-error
$\epsilon$-trajectory tracking by creating a trajectory that, when
tracked by the $\epsilon$-point, results in perfect tracking of the
original trajectory by the robot. The formation of the new trajectory
is addressed in Section \ref{subsec:Epsilon-Path-Generation} and
Section \ref{subsec:Epsilon-Trajectory-Tracking-Stability} analyzes
the stability properties of tracking the new trajectory using $\epsilon$-tracking.

\subsection{Generating an $\epsilon$-Trajectory\label{subsec:Epsilon-Path-Generation}}

It was shown in Lemma \ref{lem:converge_unicycle} that $\epsilon$-tracking
produces a steady-state error equal to the length of $\epsilon$.
To eliminate the steady-state error, this section develops a method
to exploit the knowledge of the desired trajectory and vehicle dynamics
to generate a new trajectory that is intended for the control point
used in $\epsilon$-tracking. The new trajectory is referred to as
an $\epsilon$-trajectory because it is generated for the $\epsilon$-point.

Let $\underbar{x}_{r}=\left[x_{r},y_{r}\right]^{T}\in\mathbb{C}^{2}$
be a reference trajectory intended for execution by a unicycle in
(\ref{eq:unicycle_dynamics}). Appendix \ref{app:Calculating-System-States}
derives the states and inputs of the unicycle model to achieve perfect
tracking without disturbances such that: 
\begin{align}
\begin{split}\psi_{r} & =\text{atan2}\left(\dot{y}_{r},\dot{x}_{r}\right)\\
v_{r} & =\sqrt{\dot{x}_{r}^{2}+\dot{y}_{r}^{2}}\\
a_{r} & =\left(\dot{x}_{r}\ddot{x}_{r}+\dot{y}_{r}\ddot{y}_{r}\right)v_{r}^{-1}\\
\omega_{r} & =\left(\dot{x}_{r}\ddot{y}_{r}-\dot{y}_{r}\ddot{x}_{r}\right)v_{r}^{-2}\\
\alpha_{r} & =\left(\dot{x}_{r}\dddot{y}_{r}-\dot{y}_{r}\dddot{x}_{r}\right)v_{r}^{-2}-2a_{r}\omega_{r}v_{r}^{-1}
\end{split}
.\label{eq:reference_trajectory_flat_states}
\end{align}
Note that if the vehicle were to achieve perfect tracking, the $\epsilon$-point
would be located at a distance of $\epsilon$ in front of the vehicle,
as depicted in Figure \ref{fig:eps-traj-generation} (a). Figure \ref{fig:eps-traj-generation}
(a) shows a vehicle driving a cosine trajectory, and the $\epsilon$-trajectory
is depicted as the point leading the trajectory.


\begin{figure*}[t]
    \centering
    \tiny
    \begin{subfigure}{.5\textwidth}
        \centering
        \def\svgwidth{1\linewidth}
        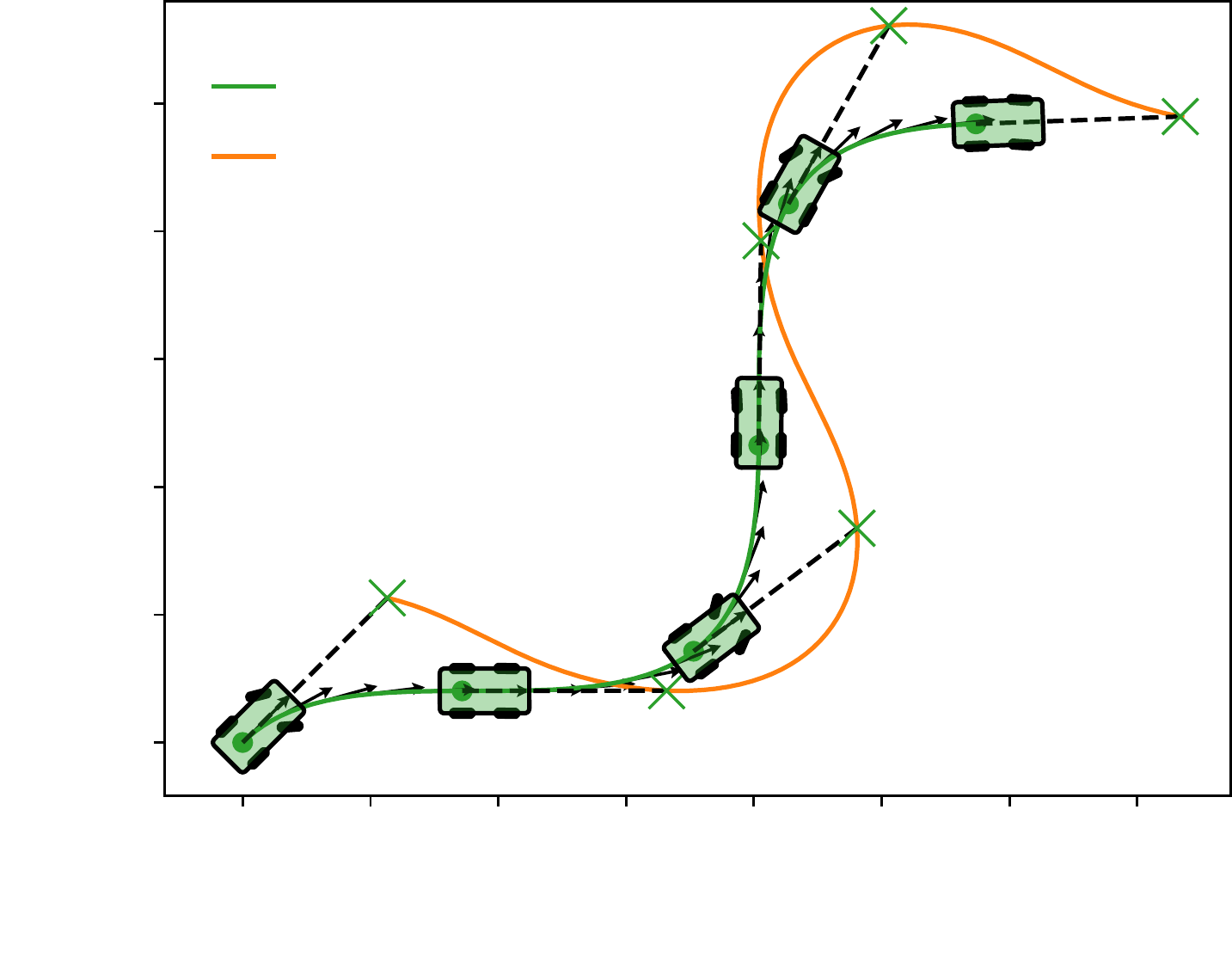
        \caption{}
    \end{subfigure}%
    \begin{subfigure}{.5\textwidth}
        \centering
        \def\svgwidth{1\linewidth}
        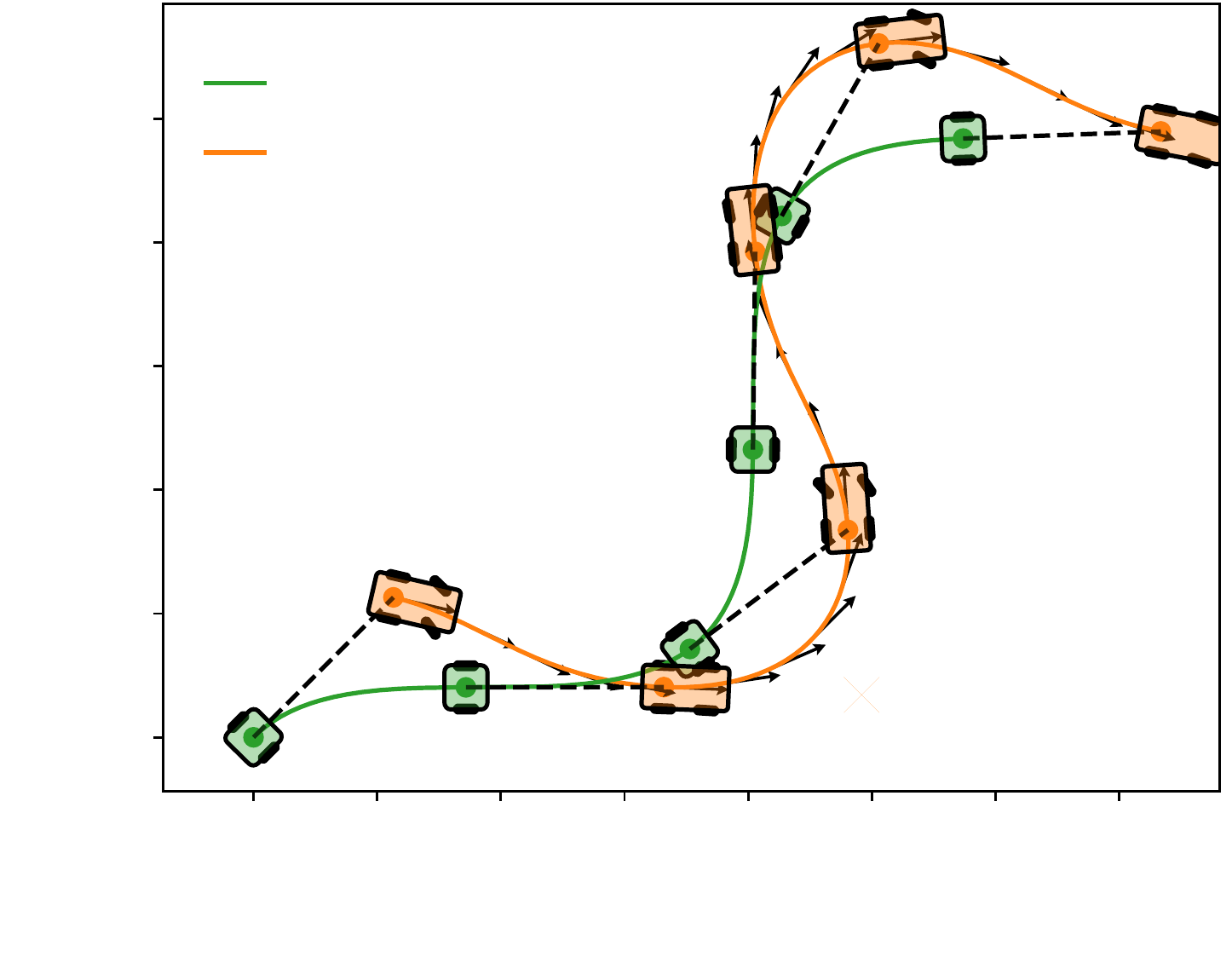
        \caption{}
    \end{subfigure}
    \caption[None]{(a) Depiction of the generation of an $\epsilon$-trajectory from the point $\underbar{q}_{\epsilon r}$, which leads the reference trajectory by a distance of $\epsilon$.
 (b) Representation of the vehicle reference point acting as a trailer being pulled by the $\epsilon$-trajectory.}
    \label{fig:eps-traj-generation}
\end{figure*}

	Thus, the $\epsilon$-trajectory, denoted $\underbar{q}_{\epsilon r}$,
is defined in terms of the reference position and orientation as:
\begin{align}
\underbar{q}_{\epsilon r}=\begin{bmatrix}x_{\epsilon r}\\
y_{\epsilon r}
\end{bmatrix} & =\begin{bmatrix}x_{r}\\
y_{r}
\end{bmatrix}+\epsilon\begin{bmatrix}\cos\psi_{r}\\
\sin\psi_{r}
\end{bmatrix}.\label{eq:eps-trajectory}
\end{align}
Subsequent derivatives can then be taken of $\underbar{q}_{\epsilon r}$
so that the controller presented in the sequel can track $\underbar{q}_{\epsilon r}$.
The first and second derivatives of $\underbar{q}_{\epsilon r}$ can
be stated as 
\begin{equation}
\begin{split}\dot{\underbar{q}}_{\epsilon r} & =R_{\epsilon r}\underbar{v}_{r}\\
\ddot{\underbar{q}}_{\epsilon r} & =R_{\epsilon r}\hat{\omega}_{r}\underbar{v}_{r}+R_{\epsilon r}\underbar{a}_{r}
\end{split}
\label{eq:reference_trajectory}
\end{equation}
where matrices are defined equivalently to (\ref{eq:eps_matrix_definitions}).
Though not necessary for trajectory tracking, the relationship between
the $\epsilon$-trajectory and the reference trajectory can be viewed
as a vehicle pulling a trailer instead of a point leading a vehicle,
as depicted in Figure \ref{fig:eps-traj-generation} (b). Figure \ref{fig:eps-traj-generation}
(b) shows that the $\epsilon$-trajectory can be represented as a
vehicle with a heading and velocities calculated similar to (\ref{eq:reference_trajectory_flat_states})
and written as: 
\begin{align}
\begin{split}\psi_{\epsilon r} & =\text{atan2}\left(\dot{y}_{\epsilon r},\dot{x}_{\epsilon r}\right)\\
v_{\epsilon r} & =\sqrt{\dot{x}_{\epsilon r}^{2}+\dot{y}_{\epsilon r}^{2}}\\
\omega_{\epsilon r} & =\left(\dot{x}_{\epsilon r}\ddot{y}_{\epsilon r}-\dot{y}_{\epsilon r}\ddot{x}_{\epsilon r}\right)v_{\epsilon r}^{-2}
\end{split}
\label{eq:y_eps_r_heading_vel_omega_r}
\end{align}
This relationship is significant because it is used to prove convergence
in the following section, and it provides additional intuition as
to how the length of $\epsilon$ and velocity of the trajectory affect
the speed of convergence.

\subsection{Stability of Zero-Error $\epsilon$-Trajectory Tracking\label{subsec:Epsilon-Trajectory-Tracking-Stability}}

With a modified trajectory for the $\epsilon$-point in hand, this
section presents a theorem and proof showing that using the $\epsilon$-tracking
controller to follow the $\epsilon$-trajectory will cause the vehicle
reference point to converge to the original reference trajectory.
The proof of the result can intuitively be thought of in two parts:
a driving phase and a pulling phase, as depicted in Figure \ref{fig:proof_stability_explanation}
(a). In the driving phase, the $\epsilon$-point converges to the
$\epsilon$-trajectory according to the gains of the control law proposed
in Lemma \ref{lem:expo-point-control}. The driving phase is depicted
in Figure \ref{fig:proof_stability_explanation} (a) as the $\epsilon$-point
aggressively connects to the $\epsilon$-trajectory, wherein the pulling
phase is entered\footnote{Note that there is no physical switching between the two phases, as
both phases occur simultaneously to produce asymptotic convergence}. In the pulling phase, the vehicle reference point and the reference
trajectory follow the $\epsilon$-trajectory at a fixed distance of
$\epsilon$, and therefore behave as two trailers being pulled by
a single vehicle with a hitch at $\underbar{q}_{\epsilon r}$. This
relationship is depicted in Figure \ref{fig:proof_stability_explanation}
(b) where the vehicle reference point and reference trajectory are
shown as two trailers being pulled by the $\epsilon$-trajectory.

The movement of the ``trailers'' are analyzed, and it is shown that
the vehicle asymptotically converges to the desired trajectory. These
phases are mirrored in the proof where Lemma \ref{lem:converge_unicycle}
is used to show the convergence of the $\epsilon$-point, and LaSalle's
invariance principle \cite{lasalle1961} is used to evaluate the invariant
set consisting of the allowable states once the driving phase has
``completed''.


\begin{figure*}[t]
    \centering
    \scriptsize
    \begin{subfigure}{.5\textwidth}
        \centering
        \def\svgwidth{1\linewidth}
        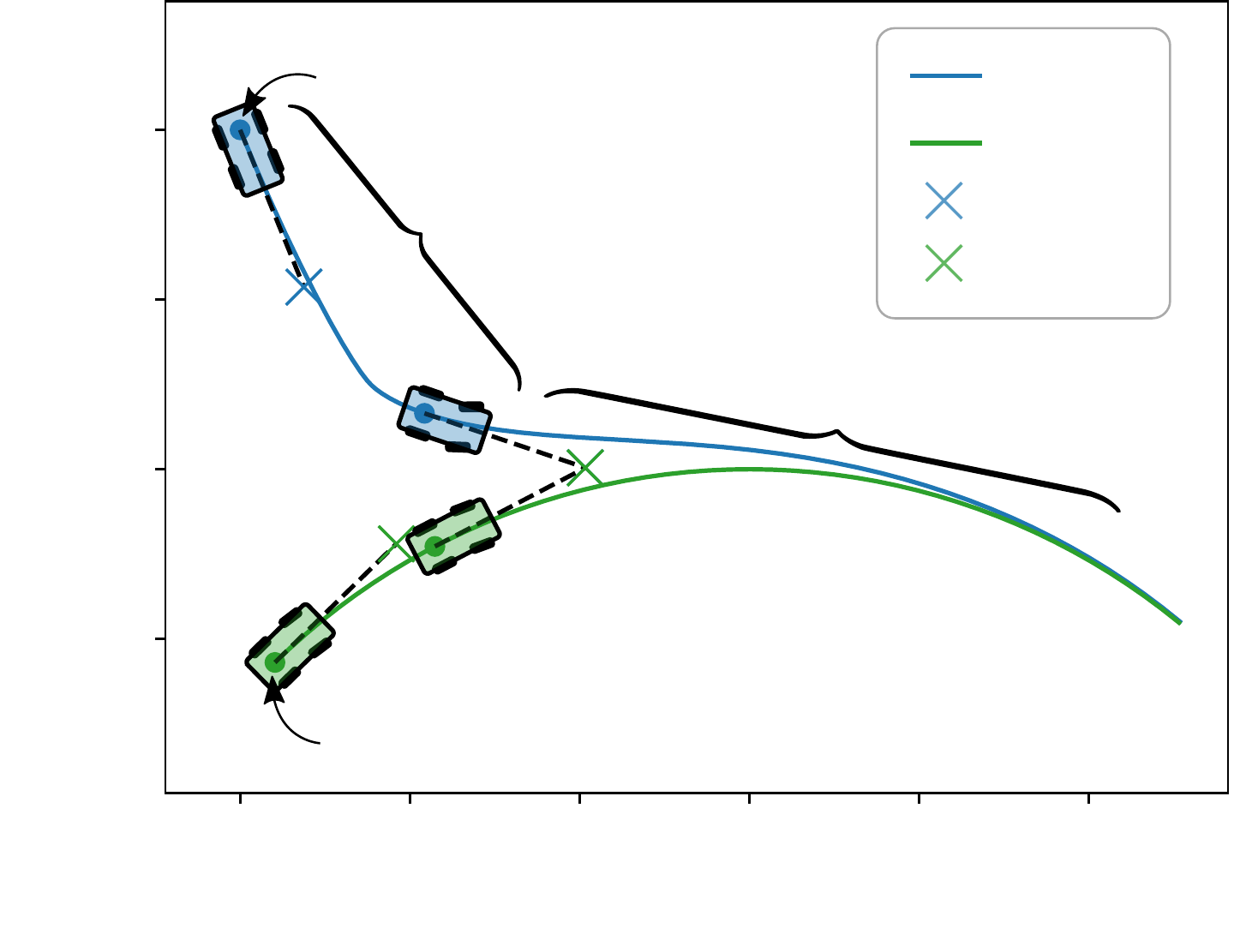
        \caption{}
    \end{subfigure}%
    \begin{subfigure}{.5\textwidth}
        \centering
        \def\svgwidth{1\linewidth}
        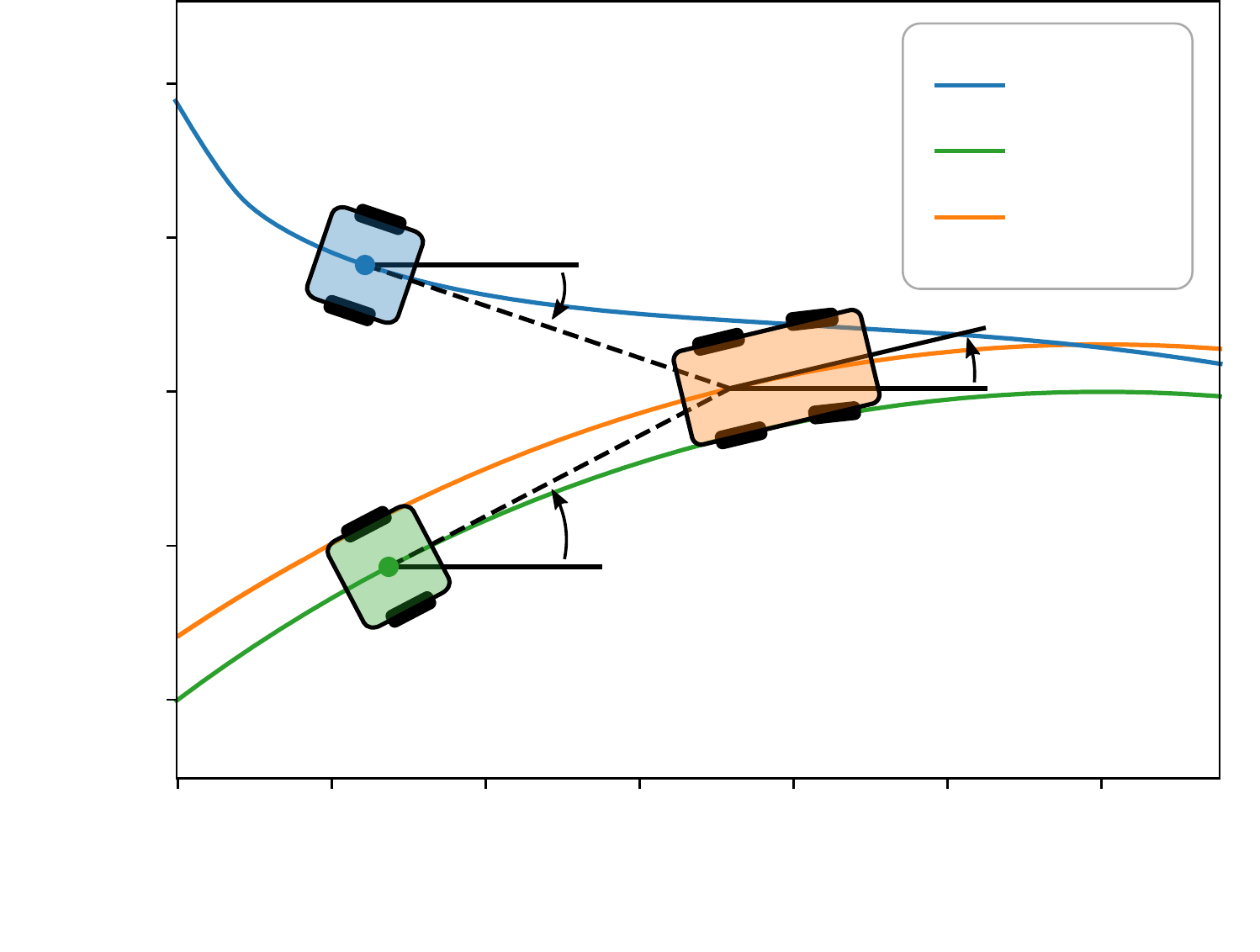
        \caption{}
    \end{subfigure}
    \caption[None]{(a) Representation of the driving and pulling phases, where the blue and green vehicles represent discrete time intervals of $\underbar{x}$ and $\underbar{x}_{r}$ respectively. The driving phase dominates as $\underbar{q}_{\epsilon}$ converges exponentially fast to $\underbar{q}_{\epsilon r}$ as shown in Lemma \ref{lem:converge_unicycle}, and the pulling phase dominates as $\underbar{q}_{\epsilon}(t) \rightarrow \underbar{q}_{\epsilon r}(t)$.
 (b) Zoomed in illustration of Figure \ref{fig:proof_stability_explanation} a during the pulling phase, which depicts the control system as two trailers being pulled by the $\epsilon$-trajectory. As $\underbar{q}_{\epsilon}(t) \rightarrow \underbar{q}_{\epsilon r}(t)$ each system follows the $\epsilon$-trajectory by a fixed value of $\epsilon$ just as trailers are pulled behind a vehicle by a hitch.}
    \label{fig:proof_stability_explanation}
\end{figure*}

The dynamic equations for a vehicle pulling two trailers at a distance
of $\epsilon$ can be expressed using (\ref{eq:trailer_equations})
by adding an additional trailer heading and substituting $\epsilon$
for $d$. Expressed in terms of the $\epsilon$-trajectory, this can
be described as: 
\begin{align}
\left[\begin{array}{c}
\dot{x}_{\epsilon r}\\
\dot{y}_{\epsilon r}\\
\dot{\psi}_{\epsilon r}\\
\dot{\psi}\\
\dot{\psi}_{r}\\
\dot{v}_{\epsilon}\\
\dot{\omega}_{\epsilon}
\end{array}\right] & =\left[\begin{array}{c}
v_{\epsilon r}\cos\left(\psi_{\epsilon r}\right)\\
v_{\epsilon r}\sin\left(\psi_{\epsilon r}\right)\\
\omega_{\epsilon}\\
\frac{v_{\epsilon r}}{\epsilon}\sin\left(\psi_{\epsilon r}-\psi\right)\\
\frac{v_{\epsilon r}}{\epsilon}\sin\left(\psi_{\epsilon r}-\psi_{r}\right)\\
0\\
0
\end{array}\right]+\left[\begin{array}{cc}
0 & 0\\
0 & 0\\
0 & 0\\
0 & 0\\
0 & 0\\
1 & 0\\
0 & 1
\end{array}\right]\left[\begin{array}{c}
a_{\epsilon}\\
\alpha_{\epsilon}
\end{array}\right]\label{eq:trailer_dynamics}
\end{align}
where the definitions for the variables in (\ref{eq:trailer_dynamics})
are listed in Table \ref{tab:paper_variables} and the two-trailer
relationship is depicted in Figure \ref{fig:proof_stability_explanation}
(b). The main result for convergence is presented in the following
theorem:	
\begin{theorem}
\label{thm:eps-path-tracking}Let $\underbar{x}_{r}$ be a reference
trajectory defined for $t\geq0$, then $\underbar{x}(t)$ will asymptotically
converge to $\underbar{x}_{r}(t)$ under the following assumptions:

\begin{enumerate}[label=\roman*)] 
	\item $\underbar{x}_{r}\in\mathbb{C}^{2}$ \item $v_r > 0$ as calculated in (\ref{eq:reference_trajectory_flat_states}) 
	\item $q_{\epsilon r}$ is defined as in (\ref{eq:reference_trajectory_flat_states}) and (\ref{eq:reference_trajectory}) 
	\item The vehicle dynamics are algebraically mapped to the unicycle model
	\item The control law from Lemma \ref{lem:converge_unicycle} is applied using $q_{\epsilon r}$ as a reference trajectory 
\end{enumerate} 

\tableTop{t}{{ l l  }
\toprule 
Variable & Description \\
\toprule 
$\underbar{x}$	  				&   Vehicle position \\
$\underbar{q}_{\epsilon}$	       &   $\epsilon$-control point in front of $\underbar{x}$ \\
$\underbar{x}_{r}$	  			&   Reference trajectory position \\
$\underbar{q}_{\epsilon r}$	     &   $\epsilon$-trajectory point in front of $\underbar{x}_{r}$ \\
$\epsilon$ 						 &   Distance from vehicle to control-reference point\\
$\psi$ 				  		   &   Heading of vehicle \\
$\psi_{r}$ 				  	   &   Heading of reference trajectory \\
$\psi_{\epsilon r}$ 				&   Heading of $\epsilon$-trajectory \\
$v_{\epsilon r}$					&   Longitudinal velocity of $\epsilon$-trajectory \\
$\omega_{\epsilon r}$ 			  &   Angular velocity of the $\epsilon$-trajectory \\
$a_{\epsilon r}$					&   Longitudinal acceleration of $\epsilon$-trajectory \\ 
$\alpha_{\epsilon r}$			   &   Anglar acceleration of $\epsilon$-trajectory \\ 

\bottomrule 
}
{\caption{Descriptions for frequently used variables.}\label{tab:paper_variables}} 	
\end{theorem}
\begin{proof}
Lemma \ref{lem:converge_unicycle} can be invoked to show that $\underbar{q}_{\epsilon}(t)\rightarrow\underbar{q}_{\epsilon r}(t)$
at an exponential rate. Due to LaSalle's invariance principle, we
need only evaluate the evolution of the state in the invariant set.
In the invariant set, 
\begin{equation}
\mathbf{0}=\underbar{q}_{\epsilon}-\underbar{q}_{\epsilon r}.\label{eq:q_epsilon_convergence}
\end{equation}
Defining $\underbar{e}=\underbar{x}_{r}-\underbar{x}$, (\ref{eq:q_epsilon_convergence})
can be re-arranged to obtain 
\begin{align}
\underbar{e} & =\begin{bmatrix}\epsilon\left(\cos\psi_{r}-\cos\psi\right)\\
\epsilon\left(\sin\psi_{r}-\sin\psi\right)
\end{bmatrix}.
\end{align}
To show that $\underbar{e}\left(t\right)\rightarrow0$ in the invariant
set, consider the state $e=\psi-\psi_{r}$. The time derivative of
$e$ can be written using (\ref{eq:trailer_dynamics}) as 
\begin{align}
\dot{e} & =-\frac{v_{\epsilon r}}{\epsilon}\left(\sin\left(\psi_{\epsilon r}-\psi_{r}\right)-\sin\left(\psi_{\epsilon r}-\psi\right)\right),
\end{align}
where $e\in\left[-\pi,\pi\right)$. Consider the Lyapunov candidate
function 
\begin{align}
V & =\frac{1}{2}e^{2}.
\end{align}
The derivative for $V$ can be expressed as
\begin{align}
\dot{V} & =e\dot{e}=-\frac{v_{\epsilon r}}{\epsilon}e\left(\sin\left(\psi_{\epsilon r}-\psi_{r}\right)-\sin\left(\psi_{\epsilon r}-\psi\right)\right).\label{eq:lyap_candidate_paper}
\end{align}
Appendix \ref{app:neg_definite_lyap} shows that
\begin{align}
\text{sign}\left(e\right) & =\text{sign}\left(\sin\left(\psi_{\epsilon r}-\psi_{r}\right)-\sin\left(\psi_{\epsilon r}-\psi\right)\right)\label{eq:sign_equality}
\end{align}
when $\left|\psi_{\epsilon r}-\psi\right|<\frac{\pi}{2}$. Since $v_{\epsilon r}$
and $\epsilon$ are strictly positive, this implies (\ref{eq:lyap_candidate_paper})
is negative definite when $\left|\psi_{\epsilon r}-\psi\right|<\frac{\pi}{2}$.
Therefore, $\psi\left(t\right)\rightarrow\psi_{r}\left(t\right)$
and $\underbar{e}\left(t\right)\rightarrow0$ asymptotically.
\end{proof}

\section{Waypoint Trajectory Planning and Following\label{sec:CCPath-Tracking}}


\begin{figure*}[t]
    \centering
    \tiny
    \begin{subfigure}{.5\textwidth}
        \centering
        \def\svgwidth{.65\linewidth}
        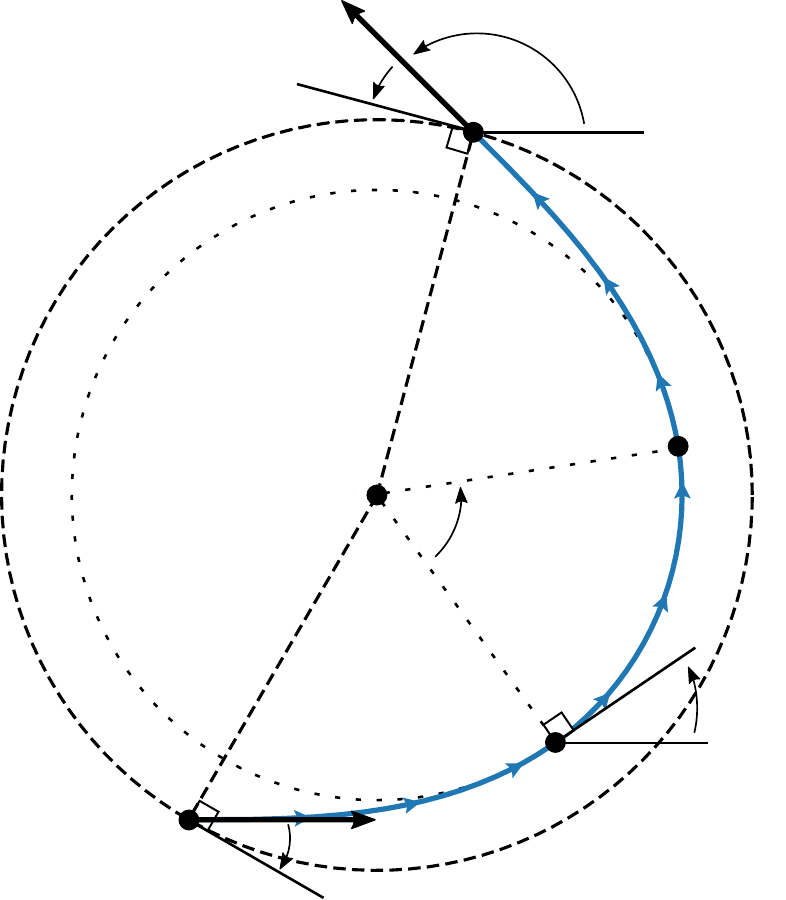
        \caption{}
    \end{subfigure}%
    \begin{subfigure}{.5\textwidth}
        \centering
        \def\svgwidth{1\linewidth}
        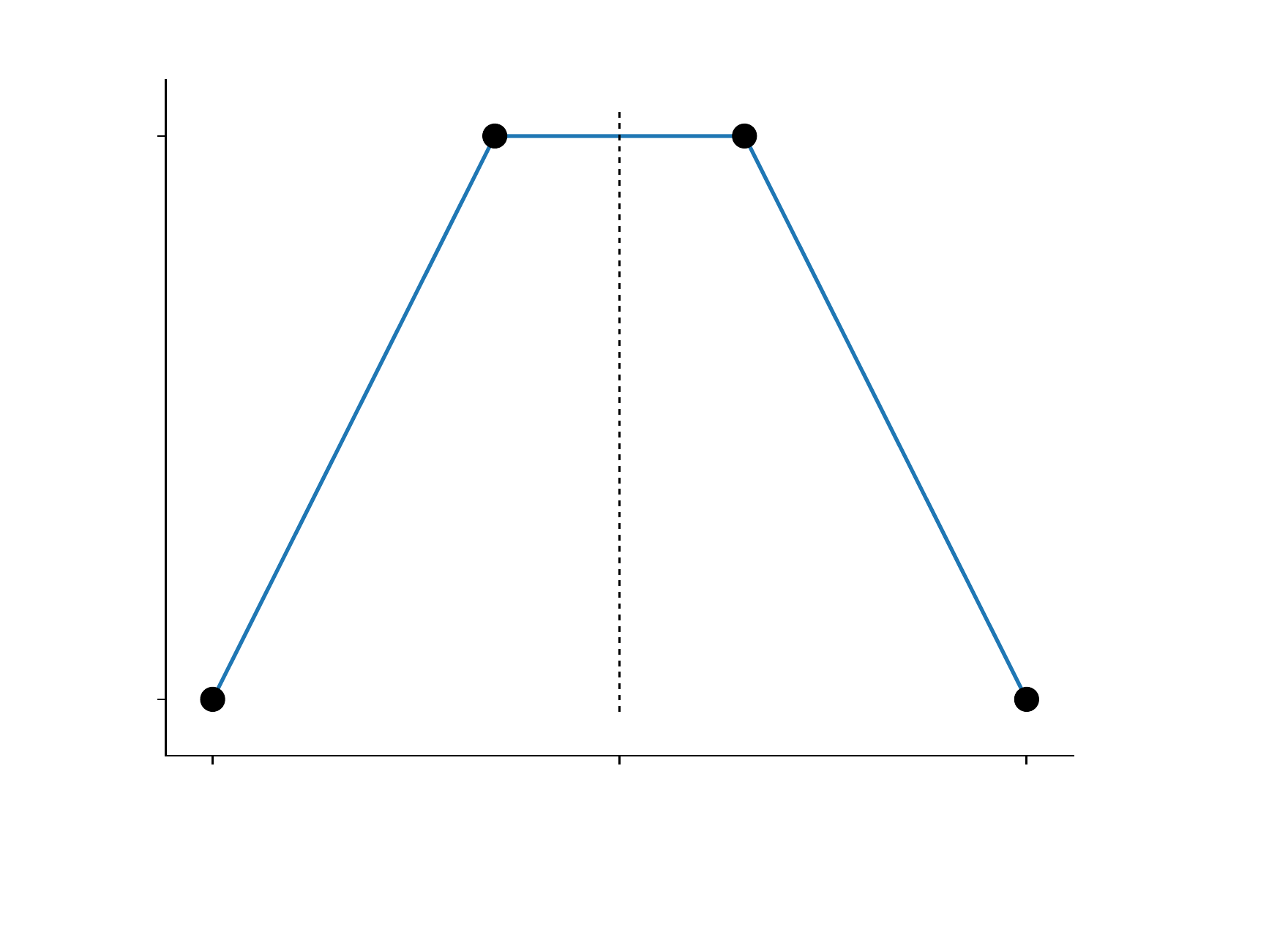
        \caption{}
    \end{subfigure}
    \caption[None]{It should be noted that though this figure is of a left turn starting with zero initial conditions, it can easily be extended to any translation or rotation and mirrored across the $x$-axis to represent a right turn. (a) Figure depicting the generation of a time-dependent CCTurn trajectory that achieves a change in heading of $\delta$ radians from starting at $w_s$ and ending at $w_e$.
 (b) Plot demonstrating the curvature continuity of a CCTurn as a function of time where $t_e$ represents the end time at the point $w_e$. Note that the curvature profile is mirrored across the midpoint of time.}
    \label{fig:CCTraj_defined}
\end{figure*}

As mentioned in the introduction, CCPaths are an established method
to plan paths that honor the curvature constraints of a wheeled vehicle.
To demonstrate the utility of the trajectory following controller
developed in previous sections, we present an extension of adding
a time index to CCPaths, and we call the result a continuous curvature
trajectory (CCTrajectory). Thus, a CCTrajectory is a time indexed
CCPath with a fixed positive velocity that connects two oriented waypoints
by concatenating turns of maximum change in curvature, referred to
as CCTurns, with straight line segments as depicted in Figure \ref{fig:clothoid_connection}.
The planning of arcs and lines in a CCTrajectory is performed in the
same way as a CCPath. As a result, the reader is referred to \cite{Fraichard2004}
for planning details, and the following is focused on CCTurns generation.

Figure \ref{fig:CCTraj_defined} shows how a time-dependent CCTurn
is developed using transition-arcs, circular-arcs, and linear paths
to achieve a desired change in heading. A CCTurn consists of four
oriented waypoints ($\omega_{s}$, $\omega_{cs}$, $\omega_{ce}$,
$\omega_{e}$) which create a change of heading equal to $\delta$,
in the following three segments:
\begin{enumerate}
\item Change from zero curvature to maximum curvature ($\omega_{s}$ to
$\omega_{cs}$)
\item Constant curvature arc ($\omega_{cs}$ to $\omega_{ce}$)
\item Change from maximum curvature to zero curvature ($\omega_{ce}$ to
$\omega_{e}$)
\end{enumerate}
where each waypoint consists of a position, orientation, starting
curvature, and linear change in curvature. In the following descriptions,
the time-indexing of each segment required to build a CCTrajectory
is explained in terms of the states in (\ref{eq:extended_dubins}),
which are used to calculate the position derivatives of the trajectory.

\subsubsection*{Change to Maximum Curvature}

To generate a time-dependent transition-arc trajectory from $w_{s}$
to $w_{cs}$ the model (\ref{eq:extended_dubins}) is integrated instead
of using (\ref{eq:fresnel_integrals}) to produce a Clothoid. Integrating
(\ref{eq:extended_dubins}) to obtain a transition arc ensures that
the spacing between points is consistent with the vehicle's desired
velocity, and it produces the necessary states at each point to solve
for the position derivatives used for vehicle control. The CCTurn
starts at $w_{s}$ with zero curvature, and it turns with the maximum
change in curvature until $\kappa=\kappa_{\text{max}}$. The corresponding
waypoints for the segment $w_{s}$ and $w_{cs}$ are
\begin{align}
w_{s}=\begin{cases}
x_{s} & \negthickspace=0\\
y & \negthickspace=0\\
\psi_{s} & \negthickspace=0\\
\kappa_{s} & \negthickspace=0\\
\sigma_{s} & \negthickspace=\sigma_{\text{max}}
\end{cases},\quad w_{cs}=\begin{cases}
x_{cs} & \negthickspace=\int_{0}^{t_{cs}}v\cos\left(\psi\left(t\right)\right)dt\\
y_{cs} & \negthickspace=\int_{0}^{t_{cs}}v\sin\left(\psi\left(t\right)\right)dt\\
\psi_{cs} & \negthickspace=\int_{0}^{t_{cs}}v\kappa\left(t\right)dt\\
\kappa_{cs} & \negthickspace=\kappa_{\text{max}}\\
\sigma_{cs} & \negthickspace=0
\end{cases}\label{eq:converge_max_curvature}
\end{align}
where the time at $w_{cs}$ is $t_{cs}=v\frac{\kappa_{max}}{\sigma_{\text{max}}}$,
and $v$ is a constant desired velocity.

\subsubsection*{Constant Curvature Arc}

After the trajectory reaches maximum curvature, it begins the creation
of a circular-arc at the point $w_{cs}$. To generate a time-dependent
circular arc (the region between $w_{cs}$ and $w_{ce}$ in Figure
\ref{fig:CCTraj_defined}), it is unnecessary to integrate (\ref{eq:extended_dubins})
to obtain the trajectory position because the closed-form equation
of a circular arc exists. The distance between points can be analytically
computed on the perimeter of the circular arc as 
\begin{align}
ds & =v\cdot dt\label{eq:trajectory_spacing}
\end{align}
where $dt$ is the discrete-time integration step, and the total circular-arc
length becomes 
\begin{align}
s & =\left(\delta-2\psi_{cs}\right)\kappa_{\text{max}}^{-1}.
\end{align}
Though (\ref{eq:extended_dubins}) does not need to be integrated
to obtain the position states of the circular trajectory, it can be
used to solve the other system states at each point. The system states
for every index on the circular arc are 
\begin{align}
\begin{split}\psi\left(t\right) & =\psi_{c}\left(t\right)+d_{c}\frac{\pi}{2}\\
\kappa\left(t\right) & =\kappa_{\text{max}},\;\sigma\left(t\right)=0\\
v\left(t\right) & =\text{const},\;\alpha\left(t\right)=0
\end{split}
\end{align}
where $\psi_{c}$ is the angle of the zero motion line referenced
from the instantaneous center of rotation and $d_{c}$ is the vehicle's
direction of motion such that a clockwise rotation is $d_{c}=1$ and
a counterclockwise rotation is $d_{c}=-1$.

\subsubsection*{Change to Zero Curvature}


\begin{figure*}[t]
    \centering
    \tiny
    \begin{subfigure}{.5\textwidth}
        \centering
        \def\svgwidth{1\linewidth}
        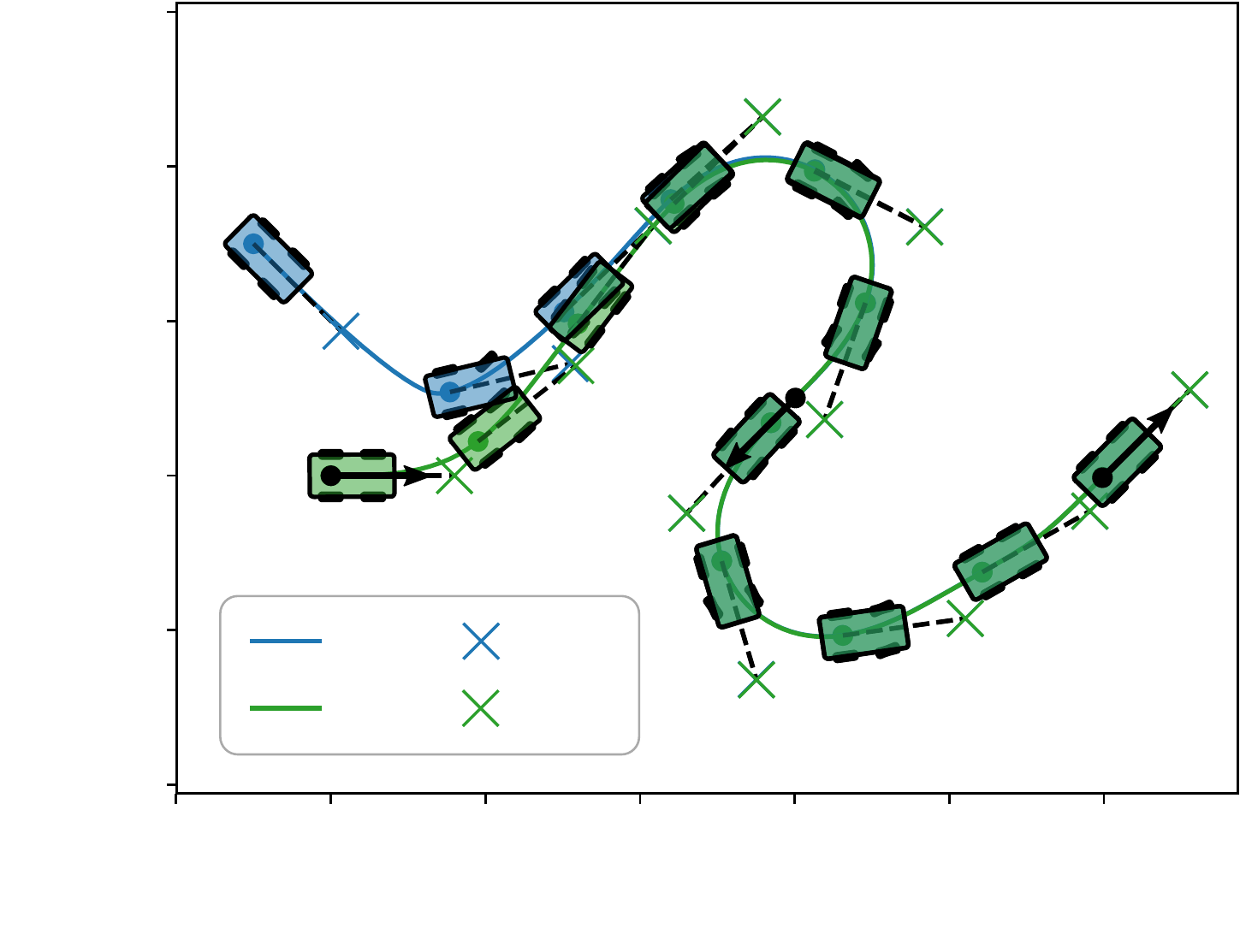
        \caption{}
    \end{subfigure}%
    \begin{subfigure}{.5\textwidth}
        \centering
        \def\svgwidth{1\linewidth}
        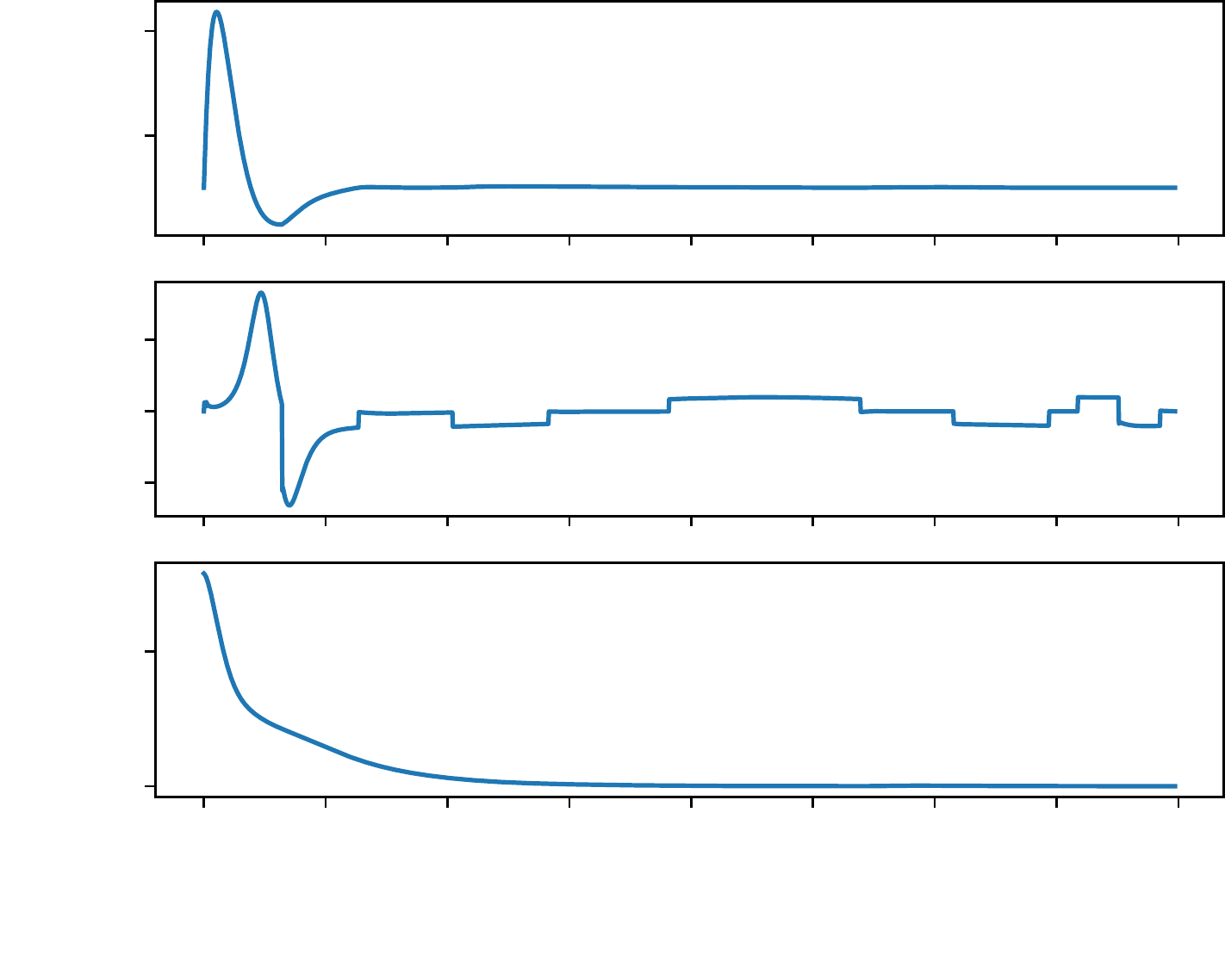
        \caption{}
    \end{subfigure}
    \caption[None]{(a) Deomonstration of zero-error $\epsilon$-trajectory tracking showing that the vehicle $\underbar{x}$ asymptotically converges to $\underbar{x}_{r}$, where $\underbar{x}_{r}$ is a CCTrajectory that passes through the waypoints $w_1$, $w_2$, and $w_3$.
 (b) Analysis of the system velocity $v$, steering rate input $\xi$, and position error between $\underbar{x}$ and $\underbar{x}_{r}$ during the demonstration of zero-error $\epsilon$-trajectory tracking.}
    \label{fig:clothoid_tracking}
\end{figure*}
As for generating the second transition arc on the CCTurn from $w_{ce}$
to $w_{e}$, it can be computed by again integrating (\ref{eq:extended_dubins})
starting with a curvature of $\kappa_{\text{max}}$ and decreasing
the curvature at the maximum rate until it reaches $0$. Equation
(\ref{eq:converge_max_curvature}) can be used in reverse as the second
transition arc is a mirror image of the first.

Finally, to generate a time-dependent linear trajectory used to connect
two CCTurns, no integration is necessary. The distance between each
point on the line is equal to (\ref{eq:trajectory_spacing}), and
the corresponding system states are 
\begin{align}
\begin{split}\psi\left(t\right) & =\text{atan2}\left(p_{2y}-p_{1y},p_{2x}-p_{1x}\right)\\
\kappa\left(t\right) & =0,\;\sigma=0\\
v\left(t\right) & =\text{const},\;\alpha=0
\end{split}
\label{eq:straight_cctraj}
\end{align}
such that $p_{1}=\left[p_{1x},p_{1y}\right]^{T}$ and $p_{2}=\left[p_{2x},p_{2y}\right]^{T}$
are the points that linearly connect two CCTurns. After the pieces
of the CCTrajectory are generated, they are concatenated to make a
continuous trajectory, and the position derivatives of the reference
trajectory are obtained by differentiating $x$ and $y$ from (\ref{eq:extended_dubins}).

\subsection{Simulation Results\label{sec:Simulation-Results}}

An example is now presented to demonstrate zero-error $\epsilon$-trajectory
tracking of a CCTrajectory. The example uses a value of $\epsilon$$\,=\,$$5\,m$
to slow the convergence of the pulling phase for the sake of illustration
but in practice the value of $\epsilon$ is system dependent. The
scenario uses a CCTrajectory with a velocity of $5\,m/s$ and three
waypoints of the form $\omega_{i}=\left[x_{i},y_{i},\text{\ensuremath{\psi}}_{i}\right]^{T}$
where $\left(x_{i},y_{i}\right)$ is the position, and $\psi$ is
the heading of the waypoint $\omega_{i}$. The waypoints are located
at $w_{1}=\left[0,0,0\right]$, $w_{2}=\left[30,5,\frac{5\pi}{4}\right]$,
and $w_{3}=\left[50,0,\frac{\pi}{4}\right]$. Figure \ref{fig:clothoid_tracking}
(a) shows the path connection between the points with a $\kappa_{\text{max}}=\pm2.7m^{-1}$
and $\sigma_{\text{max}}=\pm0.17\ \left(ms\right)^{-1}$. Figure \ref{fig:clothoid_tracking}
(b) shows the convergence to zero-error tracking.

\section{Conclusion\label{sec:Conclusion}}

In this paper, a controller for kinematic vehicles with acceleration
inputs and first-order nonholonomic constraints has been presented
that results in zero-error trajectory tracking. The controller uses
partial feedback linearization to control a new reference point in
front of the vehicle's rear axle. However, controlling a new reference
point results in a steady-state error between the original reference
point and the desired reference trajectory. Therefore, the generation
of a new trajectory is proposed that when tracked by the new reference
point, results in zero-error vehicle tracking of the original reference
trajectory. It is proven that any trajectory that is twice differentiable
guarantees zero-error tracking using the proposed controller when
the vehicle is at least within $\pm\frac{\pi}{2}$ radians of the
direction of the new trajectory. This paper further demonstrated how
to generate time-dependent, continuous curvature trajectories that
connect waypoints, given a maximum curvature and change in curvature.
The presented controller is proven to asymptotically track these trajectories,
and the results are demonstrated through simulation.

\bibliographystyle{spmpsci}
\bibliography{bibliography.bib}

\appendix\normalsize

\section{Calculating System States for a Unicycle Robot Following a Reference
Trajectory \label{app:Calculating-System-States}}

Given a reference trajectory at a specific point in time, $\underbar{x}_{r}(t)$,
this appendix derives the corresponding unicycle state that would
perfectly follow the reference.

The velocity can be found using the first two lines in (\ref{eq:unicycle_dynamics}).
Note that $\begin{bmatrix}\cos(\psi) & \sin(\psi)\end{bmatrix}^{T}$
is a unit vector. Assuming $v\geq0$, $v$ is simply the magnitude
of the position derivatives: 
\begin{equation}
v=\left\Vert \begin{bmatrix}\dot{x}\\
\dot{y}
\end{bmatrix}\right\Vert =\sqrt{\dot{x}^{2}+\dot{y}^{2}}\label{eq:app:velocity}
\end{equation}
The longitudinal acceleration can be found by directly differentiating
(\ref{eq:app:velocity}) to obtain 
\begin{equation}
a=\dot{v}=\frac{d}{dt}\left(\dot{x}^{2}+\dot{y}^{2}\right)^{1/2}=\left(\dot{x}\ddot{x}+\dot{y}\ddot{y}\right)v^{-1}
\end{equation}
As the heading is the direction of motion, it can be found using the
arc-tangent of the velocity vector: 
\begin{align}
\psi & =\text{atan2}\left(\dot{y},\dot{x}\right).\label{eq:heading}
\end{align}
Differentiating (\ref{eq:heading}), the angular velocity and acceleration
can be found as 
\begin{align}
\omega & =\frac{d}{dt}\left(\text{tan}^{-1}\left(\frac{\dot{y}}{\dot{x}}\right)\right)=\left(\dot{x}\ddot{y}-\dot{y}\ddot{x}\right)v^{-2}
\end{align}
\begin{equation}
\begin{split}\alpha & =\frac{d}{dt}\Bigl(\dot{x}\ddot{y}-\dot{y}\ddot{x}\Bigr)v^{-2}+(\dot{x}\ddot{y}-\dot{y}\ddot{x})\frac{d}{dt}\Bigl(v^{-2}\Bigr)\\
 & =(\dot{x}y^{(3)}-\dot{y}x^{(3)})v^{-2}-2(\dot{x}\ddot{y}-\dot{y}\ddot{x})v^{-3}\dot{v}\\
 & =(\dot{x}y^{(3)}-\dot{y}x^{(3)})v^{-2}-2a\omega v^{-1}
\end{split}
\end{equation}
which gives the results presented in (\ref{eq:reference_trajectory_flat_states}).

\section{Negative Definite Lyapunov Candidate\label{app:neg_definite_lyap}}

This appendix shows that (\ref{eq:sign_equality}) holds when $\left|\psi_{\epsilon r}-\psi\right|<\frac{\pi}{2}$.
Let $e_{\psi}=\psi_{\epsilon r}-\psi$ and $e_{\psi_{r}}=\psi_{\epsilon r}-\psi_{r}$,
and let 
\begin{align}
e & =\psi-\psi_{r}=e_{\psi_{r}}-e_{\psi},
\end{align}
which allows (\ref{eq:sign_equality}) to be written as 
\begin{align}
\text{sign}\left(e\right) & =\text{sign}\left(\sin\left(e_{\psi_{r}}\right)-\sin\left(e_{\psi}\right)\right).\label{eq:sign_equality_rewritten}
\end{align}

Note that sine is monotonically increasing on the interval $\begin{bmatrix}-\frac{\pi}{2} & \frac{\pi}{2}\end{bmatrix}$.
Thus, to show that (\ref{eq:sign_equality_rewritten}) holds when
$\left|e_{\psi}\right|<\frac{\pi}{2}$, it is sufficient to show that
$\left|e_{\psi_{r}}\right|<\frac{\pi}{2}$.

The angle $\psi_{\epsilon r}$ is defined as $\psi_{\epsilon r}=\text{atan2}\left(\dot{y}_{\epsilon r},\dot{x}_{\epsilon r}\right)$
where 
\begin{align}
\begin{split}\dot{\underbar{q}}_{\epsilon r}=R_{\epsilon r}\underbar{v}_{r} & =\begin{bmatrix}\begin{array}{cc}
\cos\psi_{r} & -\epsilon\sin\psi_{r}\\
\sin\psi_{r} & \epsilon\cos\psi_{r}
\end{array}\end{bmatrix}\begin{bmatrix}v_{r}\\
\omega_{r}
\end{bmatrix}\\
 & =\begin{bmatrix}\cos\psi_{r}\\
\sin\psi_{r}
\end{bmatrix}v_{r}+\begin{bmatrix}-\sin\psi_{r}\\
\cos\psi_{r}
\end{bmatrix}\epsilon\omega_{r}
\end{split}
\end{align}
Since $v$ is constrained to be positive, the direction vector that
produces $\psi_{\epsilon r}$ has a component pointing in the same
direction as $\psi_{r}$. The $\epsilon\omega_{r}$ term points in
a direction orthogonal to that of $\psi_{r}$. However, the linear
combination of the two vectors will produce a vector oriented in the
half-plane defined by $\psi_{r}$, thus $\left|\psi_{\epsilon r}-\psi_{r}\right|<\frac{\pi}{2}$
with $\left|\psi_{\epsilon r}-\psi_{r}\right|\rightarrow\frac{\pi}{2}$
as $v\rightarrow0$ or $\omega\rightarrow\infty$ and $\left|e_{\psi_{r}}\right|<\frac{\pi}{2}$
holds true.

%% file: fig_trailer_dynamics.pdf_tex
\begingroup%
  \makeatletter%
  \providecommand\color[2][]{%
    \errmessage{(Inkscape) Color is used for the text in Inkscape, but the package 'color.sty' is not loaded}%
    \renewcommand\color[2][]{}%
  }%
  \providecommand\transparent[1]{%
    \errmessage{(Inkscape) Transparency is used (non-zero) for the text in Inkscape, but the package 'transparent.sty' is not loaded}%
    \renewcommand\transparent[1]{}%
  }%
  \providecommand\rotatebox[2]{#2}%
  \newcommand*\fsize{\dimexpr\f@size pt\relax}%
  \newcommand*\lineheight[1]{\fontsize{\fsize}{#1\fsize}\selectfont}%
  \ifx\svgwidth\undefined%
    \setlength{\unitlength}{357.28610862bp}%
    \ifx\svgscale\undefined%
      \relax%
    \else%
      \setlength{\unitlength}{\unitlength * \real{\svgscale}}%
    \fi%
  \else%
    \setlength{\unitlength}{\svgwidth}%
  \fi%
  \global\let\svgwidth\undefined%
  \global\let\svgscale\undefined%
  \makeatother%
  \begin{picture}(1,0.53766596)%
    \lineheight{1}%
    \setlength\tabcolsep{0pt}%
    \put(0,0){\includegraphics[width=\unitlength,page=1]{fig_trailer_dynamics.pdf}}%
    \put(0.5639253,0.04473527){\makebox(0,0)[lt]{\lineheight{917.5}\smash{\begin{tabular}[t]{l}$\psi$\end{tabular}}}}%
    \put(0.44327061,0.3968867){\makebox(0,0)[lt]{\lineheight{917.5}\smash{\begin{tabular}[t]{l}$\psi_t$\end{tabular}}}}%
    \put(0.42139581,0.30506183){\rotatebox{-23.91940168}{\makebox(0,0)[t]{\lineheight{917.5}\smash{\begin{tabular}[t]{c}$d$\end{tabular}}}}}%
    \put(0.91394272,0.45424152){\makebox(0,0)[lt]{\lineheight{917.5}\smash{\begin{tabular}[t]{l}$\phi$\end{tabular}}}}%
    \put(0.60041058,0.41889365){\rotatebox{44.48309659}{\makebox(0,0)[t]{\lineheight{917.5}\smash{\begin{tabular}[t]{c}L\end{tabular}}}}}%
    \put(0.07243188,0.2868581){\makebox(0,0)[t]{\lineheight{917.5}\smash{\begin{tabular}[t]{c}$(x_t,y_t)$\end{tabular}}}}%
    \put(0.86264139,0.12779219){\makebox(0,0)[t]{\lineheight{917.5}\smash{\begin{tabular}[t]{c}$(x,y)$\end{tabular}}}}%
  \end{picture}%
\endgroup%

%% file: fig_bicycle_eps.pdf_tex
\begingroup%
  \makeatletter%
  \providecommand\color[2][]{%
    \errmessage{(Inkscape) Color is used for the text in Inkscape, but the package 'color.sty' is not loaded}%
    \renewcommand\color[2][]{}%
  }%
  \providecommand\transparent[1]{%
    \errmessage{(Inkscape) Transparency is used (non-zero) for the text in Inkscape, but the package 'transparent.sty' is not loaded}%
    \renewcommand\transparent[1]{}%
  }%
  \providecommand\rotatebox[2]{#2}%
  \newcommand*\fsize{\dimexpr\f@size pt\relax}%
  \newcommand*\lineheight[1]{\fontsize{\fsize}{#1\fsize}\selectfont}%
  \ifx\svgwidth\undefined%
    \setlength{\unitlength}{547.25024126bp}%
    \ifx\svgscale\undefined%
      \relax%
    \else%
      \setlength{\unitlength}{\unitlength * \real{\svgscale}}%
    \fi%
  \else%
    \setlength{\unitlength}{\svgwidth}%
  \fi%
  \global\let\svgwidth\undefined%
  \global\let\svgscale\undefined%
  \makeatother%
  \begin{picture}(1,0.86757782)%
    \lineheight{1}%
    \setlength\tabcolsep{0pt}%
    \put(-0.28195392,0.72970025){\rotatebox{-33.70430165}{\makebox(0,0)[lt]{\begin{minipage}{0.10979328\unitlength}\raggedright \end{minipage}}}}%
    \put(-0.07676159,1.03824108){\rotatebox{-33.70430165}{\makebox(0,0)[lt]{\begin{minipage}{0.10930806\unitlength}\raggedright \end{minipage}}}}%
    \put(0,0){\includegraphics[width=\unitlength,page=1]{fig_bicycle_eps.pdf}}%
    \put(0.90615114,0.04655618){\makebox(0,0)[lt]{\lineheight{1.25}\smash{\begin{tabular}[t]{l}$x$\end{tabular}}}}%
    \put(0.07533804,0.80340586){\makebox(0,0)[t]{\lineheight{1.25}\smash{\begin{tabular}[t]{c}$y$\end{tabular}}}}%
    \put(0.57621281,0.28738722){\makebox(0,0)[lt]{\lineheight{1.25}\smash{\begin{tabular}[t]{l}$\psi$\end{tabular}}}}%
    \put(0.33411783,0.26319821){\makebox(0,0)[rt]{\lineheight{1.25}\smash{\begin{tabular}[t]{r}$(x,y)$\end{tabular}}}}%
    \put(0.67632748,0.26216076){\makebox(0,0)[t]{\lineheight{1.25}\smash{\begin{tabular}[t]{c}$L$\end{tabular}}}}%
    \put(0.86124839,0.5563537){\makebox(0,0)[t]{\lineheight{1.25}\smash{\begin{tabular}[t]{c}$\phi$\end{tabular}}}}%
    \put(0.46941728,0.34490927){\rotatebox{27.891788}{\makebox(0,0)[t]{\lineheight{1.25}\smash{\begin{tabular}[t]{c}$\epsilon$\end{tabular}}}}}%
    \put(0.57238407,0.44083281){\makebox(0,0)[t]{\lineheight{1.25}\smash{\begin{tabular}[t]{c}$\underbar{q}_{\epsilon}$\end{tabular}}}}%
  \end{picture}%
\endgroup%

%% file: fig_cc_path_example.pdf_tex
\begingroup%
  \makeatletter%
  \providecommand\color[2][]{%
    \errmessage{(Inkscape) Color is used for the text in Inkscape, but the package 'color.sty' is not loaded}%
    \renewcommand\color[2][]{}%
  }%
  \providecommand\transparent[1]{%
    \errmessage{(Inkscape) Transparency is used (non-zero) for the text in Inkscape, but the package 'transparent.sty' is not loaded}%
    \renewcommand\transparent[1]{}%
  }%
  \providecommand\rotatebox[2]{#2}%
  \newcommand*\fsize{\dimexpr\f@size pt\relax}%
  \newcommand*\lineheight[1]{\fontsize{\fsize}{#1\fsize}\selectfont}%
  \ifx\svgwidth\undefined%
    \setlength{\unitlength}{427.17655351bp}%
    \ifx\svgscale\undefined%
      \relax%
    \else%
      \setlength{\unitlength}{\unitlength * \real{\svgscale}}%
    \fi%
  \else%
    \setlength{\unitlength}{\svgwidth}%
  \fi%
  \global\let\svgwidth\undefined%
  \global\let\svgscale\undefined%
  \makeatother%
  \begin{picture}(1,0.76706929)%
    \lineheight{1}%
    \setlength\tabcolsep{0pt}%
    \put(0,0){\includegraphics[width=\unitlength,page=1]{fig_cc_path_example.pdf}}%
    \put(0.31289613,0.07038851){\color[rgb]{0,0,0}\makebox(0,0)[t]{\lineheight{1.25}\smash{\begin{tabular}[t]{c}$-10$\end{tabular}}}}%
    \put(0.49712359,0.07038851){\color[rgb]{0,0,0}\makebox(0,0)[t]{\lineheight{1.25}\smash{\begin{tabular}[t]{c}$0$\end{tabular}}}}%
    \put(0.68024768,0.07038851){\color[rgb]{0,0,0}\makebox(0,0)[t]{\lineheight{1.25}\smash{\begin{tabular}[t]{c}$10$\end{tabular}}}}%
    \put(0.86507819,0.07038851){\color[rgb]{0,0,0}\makebox(0,0)[t]{\lineheight{1.25}\smash{\begin{tabular}[t]{c}$20$\end{tabular}}}}%
    \put(0.12515933,0.1676614){\color[rgb]{0,0,0}\makebox(0,0)[rt]{\lineheight{1.25}\smash{\begin{tabular}[t]{r}$0$\end{tabular}}}}%
    \put(0.12534222,0.25961705){\color[rgb]{0,0,0}\makebox(0,0)[rt]{\lineheight{1.25}\smash{\begin{tabular}[t]{r}$5$\end{tabular}}}}%
    \put(0.12358942,0.3515713){\color[rgb]{0,0,0}\makebox(0,0)[rt]{\lineheight{1.25}\smash{\begin{tabular}[t]{r}$10$\end{tabular}}}}%
    \put(0.12377231,0.44352627){\color[rgb]{0,0,0}\makebox(0,0)[rt]{\lineheight{1.25}\smash{\begin{tabular}[t]{r}$15$\end{tabular}}}}%
    \put(0.12543155,0.5354812){\color[rgb]{0,0,0}\makebox(0,0)[rt]{\lineheight{1.25}\smash{\begin{tabular}[t]{r}$20$\end{tabular}}}}%
    \put(0.12561444,0.62743616){\color[rgb]{0,0,0}\makebox(0,0)[rt]{\lineheight{1.25}\smash{\begin{tabular}[t]{r}$25$\end{tabular}}}}%
    \put(0.12515539,0.71939109){\color[rgb]{0,0,0}\makebox(0,0)[rt]{\lineheight{1.25}\smash{\begin{tabular}[t]{r}$30$\end{tabular}}}}%
    \put(0.58111115,0.00584324){\makebox(0,0)[t]{\lineheight{1.25}\smash{\begin{tabular}[t]{c}$x$ position (m)\end{tabular}}}}%
    \put(0.0213429,0.45470316){\rotatebox{90}{\makebox(0,0)[t]{\lineheight{1.25}\smash{\begin{tabular}[t]{c}$y$ position (m)\end{tabular}}}}}%
    \put(0.52209912,0.15252399){\makebox(0,0)[lt]{\lineheight{1.25}\smash{\begin{tabular}[t]{l}$w_1$\end{tabular}}}}%
    \put(0.78646218,0.42829897){\makebox(0,0)[lt]{\lineheight{1.25}\smash{\begin{tabular}[t]{l}$w_2$\end{tabular}}}}%
    \put(0.27970918,0.68001711){\makebox(0,0)[lt]{\lineheight{1.25}\smash{\begin{tabular}[t]{l}Striaght Line\end{tabular}}}}%
    \put(0.27970918,0.6351622){\makebox(0,0)[lt]{\lineheight{1.25}\smash{\begin{tabular}[t]{l}Circular Arc\end{tabular}}}}%
    \put(0.27970918,0.59413635){\makebox(0,0)[lt]{\lineheight{1.25}\smash{\begin{tabular}[t]{l}Transition Arc\end{tabular}}}}%
  \end{picture}%
\endgroup%

%% file: fig_driving_example.pdf_tex
\begingroup%
  \makeatletter%
  \providecommand\color[2][]{%
    \errmessage{(Inkscape) Color is used for the text in Inkscape, but the package 'color.sty' is not loaded}%
    \renewcommand\color[2][]{}%
  }%
  \providecommand\transparent[1]{%
    \errmessage{(Inkscape) Transparency is used (non-zero) for the text in Inkscape, but the package 'transparent.sty' is not loaded}%
    \renewcommand\transparent[1]{}%
  }%
  \providecommand\rotatebox[2]{#2}%
  \newcommand*\fsize{\dimexpr\f@size pt\relax}%
  \newcommand*\lineheight[1]{\fontsize{\fsize}{#1\fsize}\selectfont}%
  \ifx\svgwidth\undefined%
    \setlength{\unitlength}{412.60419265bp}%
    \ifx\svgscale\undefined%
      \relax%
    \else%
      \setlength{\unitlength}{\unitlength * \real{\svgscale}}%
    \fi%
  \else%
    \setlength{\unitlength}{\svgwidth}%
  \fi%
  \global\let\svgwidth\undefined%
  \global\let\svgscale\undefined%
  \makeatother%
  \begin{picture}(1,0.78654113)%
    \lineheight{1}%
    \setlength\tabcolsep{0pt}%
    \put(0,0){\includegraphics[width=\unitlength,page=1]{fig_driving_example.pdf}}%
    \put(0.20424057,0.067133){\color[rgb]{0,0,0}\makebox(0,0)[t]{\lineheight{1.25}\smash{\begin{tabular}[t]{c}$0$\end{tabular}}}}%
    \put(0.30553395,0.067133){\color[rgb]{0,0,0}\makebox(0,0)[t]{\lineheight{1.25}\smash{\begin{tabular}[t]{c}$1$\end{tabular}}}}%
    \put(0.4118265,0.067133){\color[rgb]{0,0,0}\makebox(0,0)[t]{\lineheight{1.25}\smash{\begin{tabular}[t]{c}$2$\end{tabular}}}}%
    \put(0.51552792,0.067133){\color[rgb]{0,0,0}\makebox(0,0)[t]{\lineheight{1.25}\smash{\begin{tabular}[t]{c}$3$\end{tabular}}}}%
    \put(0.61960308,0.067133){\color[rgb]{0,0,0}\makebox(0,0)[t]{\lineheight{1.25}\smash{\begin{tabular}[t]{c}$4$\end{tabular}}}}%
    \put(0.72310772,0.067133){\color[rgb]{0,0,0}\makebox(0,0)[t]{\lineheight{1.25}\smash{\begin{tabular}[t]{c}$5$\end{tabular}}}}%
    \put(0.82683966,0.067133){\color[rgb]{0,0,0}\makebox(0,0)[t]{\lineheight{1.25}\smash{\begin{tabular}[t]{c}$6$\end{tabular}}}}%
    \put(0.93048284,0.067133){\color[rgb]{0,0,0}\makebox(0,0)[t]{\lineheight{1.25}\smash{\begin{tabular}[t]{c}$7$\end{tabular}}}}%
    \put(0.56540324,0.00510051){\color[rgb]{0,0,0}\makebox(0,0)[t]{\lineheight{1.25}\smash{\begin{tabular}[t]{c}$x$-position $(m)$\end{tabular}}}}%
    \put(0.08420853,0.17455584){\color[rgb]{0,0,0}\makebox(0,0)[rt]{\lineheight{1.25}\smash{\begin{tabular}[t]{r}$1$\end{tabular}}}}%
    \put(0.08928727,0.27831096){\color[rgb]{0,0,0}\makebox(0,0)[rt]{\lineheight{1.25}\smash{\begin{tabular}[t]{r}$2$\end{tabular}}}}%
    \put(0.08917887,0.38206392){\color[rgb]{0,0,0}\makebox(0,0)[rt]{\lineheight{1.25}\smash{\begin{tabular}[t]{r}$3$\end{tabular}}}}%
    \put(0.08981791,0.48581952){\color[rgb]{0,0,0}\makebox(0,0)[rt]{\lineheight{1.25}\smash{\begin{tabular}[t]{r}$4$\end{tabular}}}}%
    \put(0.08932088,0.58957513){\color[rgb]{0,0,0}\makebox(0,0)[rt]{\lineheight{1.25}\smash{\begin{tabular}[t]{r}$5$\end{tabular}}}}%
    \put(0.08927354,0.69332832){\color[rgb]{0,0,0}\makebox(0,0)[rt]{\lineheight{1.25}\smash{\begin{tabular}[t]{r}$6$\end{tabular}}}}%
    \put(0.01894644,0.46323516){\color[rgb]{0,0,0}\rotatebox{90}{\makebox(0,0)[t]{\lineheight{1.25}\smash{\begin{tabular}[t]{c}$y$-position $(m)$\end{tabular}}}}}%
    \put(0.23974927,0.70338645){\makebox(0,0)[lt]{\lineheight{1.25}\smash{\begin{tabular}[t]{l}Ref. trajectory: $\bar{x}_r (t) $\end{tabular}}}}%
    \put(0.23974927,0.64604098){\makebox(0,0)[lt]{\lineheight{1.25}\smash{\begin{tabular}[t]{l}$\epsilon$-trajectory: $\underbar{q}_{\epsilon r}(t)$\end{tabular}}}}%
  \end{picture}%
\endgroup%

%% file: fig_pulling_example.pdf_tex
\begingroup%
  \makeatletter%
  \providecommand\color[2][]{%
    \errmessage{(Inkscape) Color is used for the text in Inkscape, but the package 'color.sty' is not loaded}%
    \renewcommand\color[2][]{}%
  }%
  \providecommand\transparent[1]{%
    \errmessage{(Inkscape) Transparency is used (non-zero) for the text in Inkscape, but the package 'transparent.sty' is not loaded}%
    \renewcommand\transparent[1]{}%
  }%
  \providecommand\rotatebox[2]{#2}%
  \newcommand*\fsize{\dimexpr\f@size pt\relax}%
  \newcommand*\lineheight[1]{\fontsize{\fsize}{#1\fsize}\selectfont}%
  \ifx\svgwidth\undefined%
    \setlength{\unitlength}{416.20918555bp}%
    \ifx\svgscale\undefined%
      \relax%
    \else%
      \setlength{\unitlength}{\unitlength * \real{\svgscale}}%
    \fi%
  \else%
    \setlength{\unitlength}{\svgwidth}%
  \fi%
  \global\let\svgwidth\undefined%
  \global\let\svgscale\undefined%
  \makeatother%
  \begin{picture}(1,0.78179474)%
    \lineheight{1}%
    \setlength\tabcolsep{0pt}%
    \put(0,0){\includegraphics[width=\unitlength,page=1]{fig_pulling_example.pdf}}%
    \put(0.21467971,0.06831295){\color[rgb]{0,0,0}\makebox(0,0)[t]{\lineheight{1.25}\smash{\begin{tabular}[t]{c}$0$\end{tabular}}}}%
    \put(0.31268015,0.06831295){\color[rgb]{0,0,0}\makebox(0,0)[t]{\lineheight{1.25}\smash{\begin{tabular}[t]{c}$1$\end{tabular}}}}%
    \put(0.41563738,0.06831295){\color[rgb]{0,0,0}\makebox(0,0)[t]{\lineheight{1.25}\smash{\begin{tabular}[t]{c}$2$\end{tabular}}}}%
    \put(0.51602351,0.06831295){\color[rgb]{0,0,0}\makebox(0,0)[t]{\lineheight{1.25}\smash{\begin{tabular}[t]{c}$3$\end{tabular}}}}%
    \put(0.61678014,0.06831295){\color[rgb]{0,0,0}\makebox(0,0)[t]{\lineheight{1.25}\smash{\begin{tabular}[t]{c}$4$\end{tabular}}}}%
    \put(0.71697369,0.06831295){\color[rgb]{0,0,0}\makebox(0,0)[t]{\lineheight{1.25}\smash{\begin{tabular}[t]{c}$5$\end{tabular}}}}%
    \put(0.81739007,0.06831295){\color[rgb]{0,0,0}\makebox(0,0)[t]{\lineheight{1.25}\smash{\begin{tabular}[t]{c}$6$\end{tabular}}}}%
    \put(0.91771853,0.06831295){\color[rgb]{0,0,0}\makebox(0,0)[t]{\lineheight{1.25}\smash{\begin{tabular}[t]{c}$7$\end{tabular}}}}%
    \put(0.08525353,0.17401677){\color[rgb]{0,0,0}\makebox(0,0)[rt]{\lineheight{1.25}\smash{\begin{tabular}[t]{r}$1$\end{tabular}}}}%
    \put(0.09028826,0.2744564){\color[rgb]{0,0,0}\makebox(0,0)[rt]{\lineheight{1.25}\smash{\begin{tabular}[t]{r}$2$\end{tabular}}}}%
    \put(0.09018081,0.37489606){\color[rgb]{0,0,0}\makebox(0,0)[rt]{\lineheight{1.25}\smash{\begin{tabular}[t]{r}$3$\end{tabular}}}}%
    \put(0.09081432,0.47533594){\color[rgb]{0,0,0}\makebox(0,0)[rt]{\lineheight{1.25}\smash{\begin{tabular}[t]{r}$4$\end{tabular}}}}%
    \put(0.09032159,0.57577578){\color[rgb]{0,0,0}\makebox(0,0)[rt]{\lineheight{1.25}\smash{\begin{tabular}[t]{r}$5$\end{tabular}}}}%
    \put(0.09027466,0.67621566){\color[rgb]{0,0,0}\makebox(0,0)[rt]{\lineheight{1.25}\smash{\begin{tabular}[t]{r}$6$\end{tabular}}}}%
    \put(0.23253593,0.70151037){\makebox(0,0)[lt]{\lineheight{1.25}\smash{\begin{tabular}[t]{l}Ref. trajectory: $\bar{x}_r (t) $\end{tabular}}}}%
    \put(0.23253593,0.6446616){\makebox(0,0)[lt]{\lineheight{1.25}\smash{\begin{tabular}[t]{l}$\epsilon$-trajectory: $\underbar{q}_{\epsilon r}(t)$\end{tabular}}}}%
    \put(0.01878234,0.45924208){\color[rgb]{0,0,0}\rotatebox{90}{\makebox(0,0)[t]{\lineheight{1.25}\smash{\begin{tabular}[t]{c}$y$-position $(m)$\end{tabular}}}}}%
    \put(0.56179578,0.00505633){\color[rgb]{0,0,0}\makebox(0,0)[t]{\lineheight{1.25}\smash{\begin{tabular}[t]{c}$x$-position $(m)$\end{tabular}}}}%
  \end{picture}%
\endgroup%

%% file: fig_proof_drive_pull.pdf_tex
\begingroup%
  \makeatletter%
  \providecommand\color[2][]{%
    \errmessage{(Inkscape) Color is used for the text in Inkscape, but the package 'color.sty' is not loaded}%
    \renewcommand\color[2][]{}%
  }%
  \providecommand\transparent[1]{%
    \errmessage{(Inkscape) Transparency is used (non-zero) for the text in Inkscape, but the package 'transparent.sty' is not loaded}%
    \renewcommand\transparent[1]{}%
  }%
  \providecommand\rotatebox[2]{#2}%
  \newcommand*\fsize{\dimexpr\f@size pt\relax}%
  \newcommand*\lineheight[1]{\fontsize{\fsize}{#1\fsize}\selectfont}%
  \ifx\svgwidth\undefined%
    \setlength{\unitlength}{418.63000586bp}%
    \ifx\svgscale\undefined%
      \relax%
    \else%
      \setlength{\unitlength}{\unitlength * \real{\svgscale}}%
    \fi%
  \else%
    \setlength{\unitlength}{\svgwidth}%
  \fi%
  \global\let\svgwidth\undefined%
  \global\let\svgscale\undefined%
  \makeatother%
  \begin{picture}(1,0.76384623)%
    \lineheight{1}%
    \setlength\tabcolsep{0pt}%
    \put(0,0){\includegraphics[width=\unitlength,page=1]{fig_proof_drive_pull.pdf}}%
    \put(0.19717289,0.0636917){\color[rgb]{0,0,0}\makebox(0,0)[t]{\lineheight{1.25}\smash{\begin{tabular}[t]{c}$-15$\end{tabular}}}}%
    \put(0.33342447,0.0636917){\color[rgb]{0,0,0}\makebox(0,0)[t]{\lineheight{1.25}\smash{\begin{tabular}[t]{c}$-10$\end{tabular}}}}%
    \put(0.47062629,0.0636917){\color[rgb]{0,0,0}\makebox(0,0)[t]{\lineheight{1.25}\smash{\begin{tabular}[t]{c}$-5$\end{tabular}}}}%
    \put(0.61289289,0.0636917){\color[rgb]{0,0,0}\makebox(0,0)[t]{\lineheight{1.25}\smash{\begin{tabular}[t]{c}$0$\end{tabular}}}}%
    \put(0.74914447,0.0636917){\color[rgb]{0,0,0}\makebox(0,0)[t]{\lineheight{1.25}\smash{\begin{tabular}[t]{c}$5$\end{tabular}}}}%
    \put(0.88445063,0.0636917){\color[rgb]{0,0,0}\makebox(0,0)[t]{\lineheight{1.25}\smash{\begin{tabular}[t]{c}$10$\end{tabular}}}}%
    \put(0.09992084,0.24179312){\color[rgb]{0,0,0}\makebox(0,0)[rt]{\lineheight{1.25}\smash{\begin{tabular}[t]{r}$15$\end{tabular}}}}%
    \put(0.10161396,0.37804494){\color[rgb]{0,0,0}\makebox(0,0)[rt]{\lineheight{1.25}\smash{\begin{tabular}[t]{r}$20$\end{tabular}}}}%
    \put(0.10180058,0.51429889){\color[rgb]{0,0,0}\makebox(0,0)[rt]{\lineheight{1.25}\smash{\begin{tabular}[t]{r}$25$\end{tabular}}}}%
    \put(0.10133216,0.65055284){\color[rgb]{0,0,0}\makebox(0,0)[rt]{\lineheight{1.25}\smash{\begin{tabular}[t]{r}$30$\end{tabular}}}}%
    \put(0.81429497,0.59028163){\makebox(0,0)[lt]{\lineheight{917.5}\smash{\begin{tabular}[t]{l}$\underbar{q}_{\epsilon}(t)$\end{tabular}}}}%
    \put(0.81349017,0.53725254){\makebox(0,0)[lt]{\lineheight{917.5}\smash{\begin{tabular}[t]{l}$\underbar{q}_{\epsilon r}(t)$\end{tabular}}}}%
    \put(0.81249875,0.6940033){\makebox(0,0)[lt]{\lineheight{917.5}\smash{\begin{tabular}[t]{l}$\underbar{x}(t)$\end{tabular}}}}%
    \put(0.81249875,0.64113101){\makebox(0,0)[lt]{\lineheight{917.5}\smash{\begin{tabular}[t]{l}$\underbar{x}_{r}(t)$\end{tabular}}}}%
    \put(0.25623843,0.69359099){\makebox(0,0)[lt]{\lineheight{917.5}\smash{\begin{tabular}[t]{l}$\underbar{x}(0)$\end{tabular}}}}%
    \put(0.26337161,0.16588192){\makebox(0,0)[lt]{\lineheight{917.5}\smash{\begin{tabular}[t]{l}$\underbar{x}_{r}(0)$\end{tabular}}}}%
    \put(0.01814886,0.44505983){\rotatebox{90}{\makebox(0,0)[t]{\lineheight{917.5}\smash{\begin{tabular}[t]{c}$y$-position $(m)$\end{tabular}}}}}%
    \put(0.55938711,0.00496877){\makebox(0,0)[t]{\lineheight{917.5}\smash{\begin{tabular}[t]{c}$x$-position $(m)$\end{tabular}}}}%
    \put(0.35637149,0.57899349){\makebox(0,0)[lt]{\lineheight{917.5}\smash{\begin{tabular}[t]{l}\small{Driving Phase}\end{tabular}}}}%
    \put(0.69752002,0.44355093){\makebox(0,0)[t]{\lineheight{917.5}\smash{\begin{tabular}[t]{c}\small{Pulling Phase}\end{tabular}}}}%
  \end{picture}%
\endgroup%

%% file: fig_proof_pull.pdf_tex
\begingroup%
  \makeatletter%
  \providecommand\color[2][]{%
    \errmessage{(Inkscape) Color is used for the text in Inkscape, but the package 'color.sty' is not loaded}%
    \renewcommand\color[2][]{}%
  }%
  \providecommand\transparent[1]{%
    \errmessage{(Inkscape) Transparency is used (non-zero) for the text in Inkscape, but the package 'transparent.sty' is not loaded}%
    \renewcommand\transparent[1]{}%
  }%
  \providecommand\rotatebox[2]{#2}%
  \newcommand*\fsize{\dimexpr\f@size pt\relax}%
  \newcommand*\lineheight[1]{\fontsize{\fsize}{#1\fsize}\selectfont}%
  \ifx\svgwidth\undefined%
    \setlength{\unitlength}{431.38930101bp}%
    \ifx\svgscale\undefined%
      \relax%
    \else%
      \setlength{\unitlength}{\unitlength * \real{\svgscale}}%
    \fi%
  \else%
    \setlength{\unitlength}{\svgwidth}%
  \fi%
  \global\let\svgwidth\undefined%
  \global\let\svgscale\undefined%
  \makeatother%
  \begin{picture}(1,0.75562131)%
    \lineheight{1}%
    \setlength\tabcolsep{0pt}%
    \put(0,0){\includegraphics[width=\unitlength,page=1]{fig_proof_pull.pdf}}%
    \put(0.82146583,0.67924412){\makebox(0,0)[lt]{\lineheight{917.5}\smash{\begin{tabular}[t]{l}$\underbar{x}(t)$\end{tabular}}}}%
    \put(0.82146583,0.62793561){\makebox(0,0)[lt]{\lineheight{917.5}\smash{\begin{tabular}[t]{l}$\underbar{x}_{r}(t)$\end{tabular}}}}%
    \put(0.82146583,0.57556632){\makebox(0,0)[lt]{\lineheight{917.5}\smash{\begin{tabular}[t]{l}$\underbar{q}_{\epsilon r}(t)$\end{tabular}}}}%
    \put(0.11671185,0.06922113){\color[rgb]{0,0,0}\makebox(0,0)[lt]{\lineheight{1.25}\smash{\begin{tabular}[t]{l}$-12$\end{tabular}}}}%
    \put(0.23893231,0.06922113){\color[rgb]{0,0,0}\makebox(0,0)[lt]{\lineheight{1.25}\smash{\begin{tabular}[t]{l}$-10$\end{tabular}}}}%
    \put(0.36852293,0.06922113){\color[rgb]{0,0,0}\makebox(0,0)[lt]{\lineheight{1.25}\smash{\begin{tabular}[t]{l}$-8$\end{tabular}}}}%
    \put(0.49074199,0.06922113){\color[rgb]{0,0,0}\makebox(0,0)[lt]{\lineheight{1.25}\smash{\begin{tabular}[t]{l}$-6$\end{tabular}}}}%
    \put(0.61296102,0.06922113){\color[rgb]{0,0,0}\makebox(0,0)[lt]{\lineheight{1.25}\smash{\begin{tabular}[t]{l}$-4$\end{tabular}}}}%
    \put(0.73518009,0.06922113){\color[rgb]{0,0,0}\makebox(0,0)[lt]{\lineheight{1.25}\smash{\begin{tabular}[t]{l}$-2$\end{tabular}}}}%
    \put(0.86710737,0.06922113){\color[rgb]{0,0,0}\makebox(0,0)[lt]{\lineheight{1.25}\smash{\begin{tabular}[t]{l}$0$\end{tabular}}}}%
    \put(0.10117617,0.19136811){\color[rgb]{0,0,0}\makebox(0,0)[rt]{\lineheight{1.25}\smash{\begin{tabular}[t]{r}$16$\end{tabular}}}}%
    \put(0.10122145,0.31358672){\color[rgb]{0,0,0}\makebox(0,0)[rt]{\lineheight{1.25}\smash{\begin{tabular}[t]{r}$18$\end{tabular}}}}%
    \put(0.10295504,0.43580812){\color[rgb]{0,0,0}\makebox(0,0)[rt]{\lineheight{1.25}\smash{\begin{tabular}[t]{r}$20$\end{tabular}}}}%
    \put(0.10284185,0.55802719){\color[rgb]{0,0,0}\makebox(0,0)[rt]{\lineheight{1.25}\smash{\begin{tabular}[t]{r}$22$\end{tabular}}}}%
    \put(0.10294372,0.68024626){\color[rgb]{0,0,0}\makebox(0,0)[rt]{\lineheight{1.25}\smash{\begin{tabular}[t]{r}$24$\end{tabular}}}}%
    \put(0.45697838,0.51531839){\makebox(0,0)[lt]{\lineheight{917.5}\smash{\begin{tabular}[t]{l}$\psi$\end{tabular}}}}%
    \put(0.464497,0.33323828){\makebox(0,0)[lt]{\lineheight{917.5}\smash{\begin{tabular}[t]{l}$\psi_r$\end{tabular}}}}%
    \put(0.78266758,0.46033364){\makebox(0,0)[lt]{\lineheight{917.5}\smash{\begin{tabular}[t]{l}$\psi_{\epsilon r}$\end{tabular}}}}%
    \put(0.01761206,0.43149573){\rotatebox{90}{\makebox(0,0)[t]{\lineheight{917.5}\smash{\begin{tabular}[t]{c}$y$-position $(m)$\end{tabular}}}}}%
    \put(0.55414281,0.00482181){\makebox(0,0)[t]{\lineheight{917.5}\smash{\begin{tabular}[t]{c}$x$-position $(m)$\end{tabular}}}}%
    \put(0.4133122,0.4660265){\rotatebox{-18.46043665}{\makebox(0,0)[t]{\lineheight{917.5}\smash{\begin{tabular}[t]{c}$\epsilon$\end{tabular}}}}}%
    \put(0.43530625,0.3890542){\rotatebox{28.22261148}{\makebox(0,0)[t]{\lineheight{917.5}\smash{\begin{tabular}[t]{c}$\epsilon$\end{tabular}}}}}%
    \put(0,0){\includegraphics[width=\unitlength,page=2]{fig_proof_pull.pdf}}%
  \end{picture}%
\endgroup%

%% file: fig_cc_turn_left.pdf_tex
\begingroup%
  \makeatletter%
  \providecommand\color[2][]{%
    \errmessage{(Inkscape) Color is used for the text in Inkscape, but the package 'color.sty' is not loaded}%
    \renewcommand\color[2][]{}%
  }%
  \providecommand\transparent[1]{%
    \errmessage{(Inkscape) Transparency is used (non-zero) for the text in Inkscape, but the package 'transparent.sty' is not loaded}%
    \renewcommand\transparent[1]{}%
  }%
  \providecommand\rotatebox[2]{#2}%
  \newcommand*\fsize{\dimexpr\f@size pt\relax}%
  \newcommand*\lineheight[1]{\fontsize{\fsize}{#1\fsize}\selectfont}%
  \ifx\svgwidth\undefined%
    \setlength{\unitlength}{226.31469498bp}%
    \ifx\svgscale\undefined%
      \relax%
    \else%
      \setlength{\unitlength}{\unitlength * \real{\svgscale}}%
    \fi%
  \else%
    \setlength{\unitlength}{\svgwidth}%
  \fi%
  \global\let\svgwidth\undefined%
  \global\let\svgscale\undefined%
  \makeatother%
  \begin{picture}(1,1.14412093)%
    \lineheight{1}%
    \setlength\tabcolsep{0pt}%
    \put(0,0){\includegraphics[width=\unitlength,page=1]{fig_cc_turn_left.pdf}}%
    \put(0.58642164,0.30767167){\makebox(0,0)[rt]{\lineheight{917.5}\smash{\begin{tabular}[t]{r}$\frac{1}{\kappa_\text{max}}$\end{tabular}}}}%
    \put(0.74218212,0.15593148){\makebox(0,0)[t]{\lineheight{917.5}\smash{\begin{tabular}[t]{c}$w_{cs}$\end{tabular}}}}%
    \put(0.2351394,0.03858115){\makebox(0,0)[t]{\lineheight{917.5}\smash{\begin{tabular}[t]{c}$w_s$\end{tabular}}}}%
    \put(0.62552584,1.00743183){\makebox(0,0)[t]{\lineheight{917.5}\smash{\begin{tabular}[t]{c}$w_e$\end{tabular}}}}%
    \put(0.13660646,0.94180017){\makebox(0,0)[rt]{\lineheight{917.5}\smash{\begin{tabular}[t]{r}Left CCTurn\end{tabular}}}}%
    \put(0.84901408,0.59080232){\makebox(0,0)[rt]{\lineheight{1.25}\smash{\begin{tabular}[t]{r}$w_{ce}$\end{tabular}}}}%
    \put(0.47236175,1.0599111){\makebox(0,0)[rt]{\lineheight{917.5}\smash{\begin{tabular}[t]{r}$\mu$\end{tabular}}}}%
    \put(0.37580834,0.05686767){\makebox(0,0)[lt]{\lineheight{917.5}\smash{\begin{tabular}[t]{l}$-\mu$\end{tabular}}}}%
    \put(0.68045404,1.09698007){\makebox(0,0)[t]{\lineheight{917.5}\smash{\begin{tabular}[t]{c}$\delta$\end{tabular}}}}%
    \put(0.59359308,0.43544571){\makebox(0,0)[lt]{\lineheight{917.5}\smash{\begin{tabular}[t]{l}$\delta-2\psi_{cs}$\end{tabular}}}}%
    \put(0.89736188,0.23903602){\makebox(0,0)[lt]{\lineheight{917.5}\smash{\begin{tabular}[t]{l}$\psi_{cs}$\end{tabular}}}}%
  \end{picture}%
\endgroup%

%% file: fig_cc_turn_time.pdf_tex
\begingroup%
  \makeatletter%
  \providecommand\color[2][]{%
    \errmessage{(Inkscape) Color is used for the text in Inkscape, but the package 'color.sty' is not loaded}%
    \renewcommand\color[2][]{}%
  }%
  \providecommand\transparent[1]{%
    \errmessage{(Inkscape) Transparency is used (non-zero) for the text in Inkscape, but the package 'transparent.sty' is not loaded}%
    \renewcommand\transparent[1]{}%
  }%
  \providecommand\rotatebox[2]{#2}%
  \newcommand*\fsize{\dimexpr\f@size pt\relax}%
  \newcommand*\lineheight[1]{\fontsize{\fsize}{#1\fsize}\selectfont}%
  \ifx\svgwidth\undefined%
    \setlength{\unitlength}{506.70951572bp}%
    \ifx\svgscale\undefined%
      \relax%
    \else%
      \setlength{\unitlength}{\unitlength * \real{\svgscale}}%
    \fi%
  \else%
    \setlength{\unitlength}{\svgwidth}%
  \fi%
  \global\let\svgwidth\undefined%
  \global\let\svgscale\undefined%
  \makeatother%
  \begin{picture}(1,0.73233949)%
    \lineheight{1}%
    \setlength\tabcolsep{0pt}%
    \put(0,0){\includegraphics[width=\unitlength,page=1]{fig_cc_turn_time.pdf}}%
    \put(0.16349179,0.08689575){\color[rgb]{0,0,0}\makebox(0,0)[t]{\lineheight{1.25}\smash{\begin{tabular}[t]{c}$0$\end{tabular}}}}%
    \put(0.79848405,0.08689575){\color[rgb]{0,0,0}\makebox(0,0)[t]{\lineheight{1.25}\smash{\begin{tabular}[t]{c}$t_e$\end{tabular}}}}%
    \put(0.10802123,0.18152681){\color[rgb]{0,0,0}\makebox(0,0)[rt]{\lineheight{1.25}\smash{\begin{tabular}[t]{r}$0$\end{tabular}}}}%
    \put(0.10817539,0.61917345){\color[rgb]{0,0,0}\makebox(0,0)[rt]{\lineheight{1.25}\smash{\begin{tabular}[t]{r}$\kappa_{\text{max}}$\end{tabular}}}}%
    \put(0.07334946,0.40781919){\color[rgb]{0,0,0}\rotatebox{90}{\makebox(0,0)[t]{\lineheight{1.25}\smash{\begin{tabular}[t]{c}Curvature $(\kappa)$\end{tabular}}}}}%
    \put(0.19302107,0.18560005){\makebox(0,0)[lt]{\lineheight{917.5}\smash{\begin{tabular}[t]{l}$w_s$\end{tabular}}}}%
    \put(0.7700612,0.18429572){\makebox(0,0)[rt]{\lineheight{917.5}\smash{\begin{tabular}[t]{r}$w_e$\end{tabular}}}}%
    \put(0.39467489,0.59504147){\makebox(0,0)[lt]{\lineheight{917.5}\smash{\begin{tabular}[t]{l}$w_{cs}$\end{tabular}}}}%
    \put(0.57418136,0.59504147){\makebox(0,0)[rt]{\lineheight{917.5}\smash{\begin{tabular}[t]{r}$w_{ce}$\end{tabular}}}}%
    \put(0.48139458,0.02302345){\color[rgb]{0,0,0}\makebox(0,0)[t]{\lineheight{1.25}\smash{\begin{tabular}[t]{c}Time $(t)$\end{tabular}}}}%
    \put(0.48150395,0.09128268){\makebox(0,0)[t]{\lineheight{917.5}\smash{\begin{tabular}[t]{c}$t_{cs}+ \frac{1}{2} \left(\delta-2\psi_{cs}\right)\kappa_{\text{max}}^{-1}v^{-1}$\end{tabular}}}}%
  \end{picture}%
\endgroup%

%% file: fig_sim_results.pdf_tex
\begingroup%
  \makeatletter%
  \providecommand\color[2][]{%
    \errmessage{(Inkscape) Color is used for the text in Inkscape, but the package 'color.sty' is not loaded}%
    \renewcommand\color[2][]{}%
  }%
  \providecommand\transparent[1]{%
    \errmessage{(Inkscape) Transparency is used (non-zero) for the text in Inkscape, but the package 'transparent.sty' is not loaded}%
    \renewcommand\transparent[1]{}%
  }%
  \providecommand\rotatebox[2]{#2}%
  \newcommand*\fsize{\dimexpr\f@size pt\relax}%
  \newcommand*\lineheight[1]{\fontsize{\fsize}{#1\fsize}\selectfont}%
  \ifx\svgwidth\undefined%
    \setlength{\unitlength}{417.070267bp}%
    \ifx\svgscale\undefined%
      \relax%
    \else%
      \setlength{\unitlength}{\unitlength * \real{\svgscale}}%
    \fi%
  \else%
    \setlength{\unitlength}{\svgwidth}%
  \fi%
  \global\let\svgwidth\undefined%
  \global\let\svgscale\undefined%
  \makeatother%
  \begin{picture}(1,0.76847148)%
    \lineheight{1}%
    \setlength\tabcolsep{0pt}%
    \put(0,0){\includegraphics[width=\unitlength,page=1]{fig_sim_results.pdf}}%
    \put(0.43456575,0.24190458){\makebox(0,0)[lt]{\lineheight{917.5}\smash{\begin{tabular}[t]{l}$\underbar{q}_{\epsilon}$\end{tabular}}}}%
    \put(0.4341696,0.18875884){\makebox(0,0)[lt]{\lineheight{917.5}\smash{\begin{tabular}[t]{l}$\underbar{q}_{\epsilon r}$\end{tabular}}}}%
    \put(0.13996674,0.06408637){\color[rgb]{0,0,0}\makebox(0,0)[t]{\lineheight{1.25}\smash{\begin{tabular}[t]{c}$-10$\end{tabular}}}}%
    \put(0.26490803,0.06408637){\color[rgb]{0,0,0}\makebox(0,0)[t]{\lineheight{1.25}\smash{\begin{tabular}[t]{c}$0$\end{tabular}}}}%
    \put(0.38888267,0.06408637){\color[rgb]{0,0,0}\makebox(0,0)[t]{\lineheight{1.25}\smash{\begin{tabular}[t]{c}$10$\end{tabular}}}}%
    \put(0.51460598,0.06408637){\color[rgb]{0,0,0}\makebox(0,0)[t]{\lineheight{1.25}\smash{\begin{tabular}[t]{c}$20$\end{tabular}}}}%
    \put(0.63924496,0.06408637){\color[rgb]{0,0,0}\makebox(0,0)[t]{\lineheight{1.25}\smash{\begin{tabular}[t]{c}$30$\end{tabular}}}}%
    \put(0.7643761,0.06408637){\color[rgb]{0,0,0}\makebox(0,0)[t]{\lineheight{1.25}\smash{\begin{tabular}[t]{c}$40$\end{tabular}}}}%
    \put(0.88881013,0.06408637){\color[rgb]{0,0,0}\makebox(0,0)[t]{\lineheight{1.25}\smash{\begin{tabular}[t]{c}$50$\end{tabular}}}}%
    \put(0.10363804,0.12582921){\color[rgb]{0,0,0}\makebox(0,0)[rt]{\lineheight{1.25}\smash{\begin{tabular}[t]{r}$-20$\end{tabular}}}}%
    \put(0.10363804,0.25060866){\color[rgb]{0,0,0}\makebox(0,0)[rt]{\lineheight{1.25}\smash{\begin{tabular}[t]{r}$-10$\end{tabular}}}}%
    \put(0.10396053,0.37538811){\color[rgb]{0,0,0}\makebox(0,0)[rt]{\lineheight{1.25}\smash{\begin{tabular}[t]{r}0\end{tabular}}}}%
    \put(0.10235258,0.50016802){\color[rgb]{0,0,0}\makebox(0,0)[rt]{\lineheight{1.25}\smash{\begin{tabular}[t]{r}$10$\end{tabular}}}}%
    \put(0.10423934,0.62494797){\color[rgb]{0,0,0}\makebox(0,0)[rt]{\lineheight{1.25}\smash{\begin{tabular}[t]{r}$20$\end{tabular}}}}%
    \put(0.10395649,0.74972792){\color[rgb]{0,0,0}\makebox(0,0)[rt]{\lineheight{1.25}\smash{\begin{tabular}[t]{r}$30$\end{tabular}}}}%
    \put(0.57091748,0.00498736){\makebox(0,0)[t]{\lineheight{917.5}\smash{\begin{tabular}[t]{c}$x$-position $(m)$\end{tabular}}}}%
    \put(0.01821673,0.44688072){\rotatebox{90}{\makebox(0,0)[t]{\lineheight{917.5}\smash{\begin{tabular}[t]{c}$y$-position $(m)$\end{tabular}}}}}%
    \put(0.2768289,0.24185285){\makebox(0,0)[lt]{\lineheight{917.5}\smash{\begin{tabular}[t]{l}$\underbar{x}$\end{tabular}}}}%
    \put(0.2768289,0.18878277){\makebox(0,0)[lt]{\lineheight{917.5}\smash{\begin{tabular}[t]{l}$\underbar{x}_{r}$\end{tabular}}}}%
    \put(0.26740304,0.33344103){\makebox(0,0)[t]{\lineheight{917.5}\smash{\begin{tabular}[t]{c}$w_1$\end{tabular}}}}%
    \put(0.63065699,0.46259177){\makebox(0,0)[rt]{\lineheight{917.5}\smash{\begin{tabular}[t]{r}$w_2$\end{tabular}}}}%
    \put(0.90236372,0.3337306){\makebox(0,0)[lt]{\lineheight{917.5}\smash{\begin{tabular}[t]{l}$w_3$\end{tabular}}}}%
  \end{picture}%
\endgroup%

%% file: fig_sim_results_states.pdf_tex
\begingroup%
  \makeatletter%
  \providecommand\color[2][]{%
    \errmessage{(Inkscape) Color is used for the text in Inkscape, but the package 'color.sty' is not loaded}%
    \renewcommand\color[2][]{}%
  }%
  \providecommand\transparent[1]{%
    \errmessage{(Inkscape) Transparency is used (non-zero) for the text in Inkscape, but the package 'transparent.sty' is not loaded}%
    \renewcommand\transparent[1]{}%
  }%
  \providecommand\rotatebox[2]{#2}%
  \newcommand*\fsize{\dimexpr\f@size pt\relax}%
  \newcommand*\lineheight[1]{\fontsize{\fsize}{#1\fsize}\selectfont}%
  \ifx\svgwidth\undefined%
    \setlength{\unitlength}{409.51875617bp}%
    \ifx\svgscale\undefined%
      \relax%
    \else%
      \setlength{\unitlength}{\unitlength * \real{\svgscale}}%
    \fi%
  \else%
    \setlength{\unitlength}{\svgwidth}%
  \fi%
  \global\let\svgwidth\undefined%
  \global\let\svgscale\undefined%
  \makeatother%
  \begin{picture}(1,0.79963137)%
    \lineheight{1}%
    \setlength\tabcolsep{0pt}%
    \put(0,0){\includegraphics[width=\unitlength,page=1]{fig_sim_results_states.pdf}}%
    \put(0.56388441,0.00358891){\makebox(0,0)[t]{\lineheight{917.5}\smash{\begin{tabular}[t]{c}Time $(s)$\end{tabular}}}}%
    \put(0.01855264,0.70309963){\rotatebox{90}{\makebox(0,0)[t]{\lineheight{917.5}\smash{\begin{tabular}[t]{c}$v\,\,(m/s)$\end{tabular}}}}}%
    \put(0.01855264,0.47375239){\rotatebox{90}{\makebox(0,0)[t]{\lineheight{917.5}\smash{\begin{tabular}[t]{c}$\xi\,\,(\text{rad}/s)$\end{tabular}}}}}%
    \put(0.01855264,0.24321281){\rotatebox{90}{\makebox(0,0)[t]{\lineheight{917.5}\smash{\begin{tabular}[t]{c}Error $(m)$\end{tabular}}}}}%
    \put(0.0912716,0.68334683){\color[rgb]{0,0,0}\makebox(0,0)[rt]{\lineheight{1.25}\smash{\begin{tabular}[t]{r}$10$\end{tabular}}}}%
    \put(0.0912716,0.76868117){\color[rgb]{0,0,0}\makebox(0,0)[rt]{\lineheight{1.25}\smash{\begin{tabular}[t]{r}$20$\end{tabular}}}}%
    \put(0.16809187,0.08754095){\color[rgb]{0,0,0}\makebox(0,0)[t]{\lineheight{1.25}\smash{\begin{tabular}[t]{c}$0.0$\end{tabular}}}}%
    \put(0.2677696,0.08754095){\color[rgb]{0,0,0}\makebox(0,0)[t]{\lineheight{1.25}\smash{\begin{tabular}[t]{c}$2.5$\end{tabular}}}}%
    \put(0.36697967,0.08754095){\color[rgb]{0,0,0}\makebox(0,0)[t]{\lineheight{1.25}\smash{\begin{tabular}[t]{c}$5.0$\end{tabular}}}}%
    \put(0.46644956,0.08754095){\color[rgb]{0,0,0}\makebox(0,0)[t]{\lineheight{1.25}\smash{\begin{tabular}[t]{c}$7.5$\end{tabular}}}}%
    \put(0.56504954,0.08754095){\color[rgb]{0,0,0}\makebox(0,0)[t]{\lineheight{1.25}\smash{\begin{tabular}[t]{c}$10.0$\end{tabular}}}}%
    \put(0.66459037,0.08754095){\color[rgb]{0,0,0}\makebox(0,0)[t]{\lineheight{1.25}\smash{\begin{tabular}[t]{c}$12.5$\end{tabular}}}}%
    \put(0.76393853,0.08754095){\color[rgb]{0,0,0}\makebox(0,0)[t]{\lineheight{1.25}\smash{\begin{tabular}[t]{c}$15.0$\end{tabular}}}}%
    \put(0.86347806,0.08754095){\color[rgb]{0,0,0}\makebox(0,0)[t]{\lineheight{1.25}\smash{\begin{tabular}[t]{c}$17.5$\end{tabular}}}}%
    \put(0.96378896,0.08754095){\color[rgb]{0,0,0}\makebox(0,0)[t]{\lineheight{1.25}\smash{\begin{tabular}[t]{c}$20.0$\end{tabular}}}}%
    \put(0.0912716,0.26178038){\color[rgb]{0,0,0}\makebox(0,0)[rt]{\lineheight{1.25}\smash{\begin{tabular}[t]{r}$10$\end{tabular}}}}%
    \put(0.0912716,0.15284258){\color[rgb]{0,0,0}\makebox(0,0)[rt]{\lineheight{1.25}\smash{\begin{tabular}[t]{r}$0$\end{tabular}}}}%
    \put(0.0912716,0.51673237){\color[rgb]{0,0,0}\makebox(0,0)[rt]{\lineheight{1.25}\smash{\begin{tabular}[t]{r}$1$\end{tabular}}}}%
    \put(0.0912716,0.45739024){\color[rgb]{0,0,0}\makebox(0,0)[rt]{\lineheight{1.25}\smash{\begin{tabular}[t]{r}$0$\end{tabular}}}}%
    \put(0.0912716,0.40147171){\color[rgb]{0,0,0}\makebox(0,0)[rt]{\lineheight{1.25}\smash{\begin{tabular}[t]{r}$-1$\end{tabular}}}}%
  \end{picture}%
\endgroup%